\documentclass[a4paper, 11pt]{article}
\usepackage{srcltx,graphicx}
\usepackage{amsmath, amssymb, amsthm}
\usepackage{color}
\usepackage{lscape}
\usepackage{multirow}
\usepackage{psfrag}
\usepackage[dvips]{hyperref}
\usepackage[hang]{subfigure}

\newtheorem{theorem}{Theorem}
\newtheorem{lemma}{Lemma}
\newtheorem{definition}{Definition}

\newcommand\bbR{\mathbb{R}}
\newcommand\bbN{\mathbb{N}}
\newcommand\bxi{\boldsymbol{\xi}}
\newcommand\bx{\boldsymbol{x}}
\newcommand\bq{\boldsymbol{q}}

\newcommand\be{\boldsymbol{e}}

\newcommand\br{\boldsymbol{r}}

\newcommand\bPhi{\boldsymbol{\Phi}}
\newcommand\dd{\,\mathrm{d}}
\newcommand\He{\mathit{He}}
\newcommand\Kn{\mathit{Kn}}
\newcommand\bw{\boldsymbol{w}}
\newcommand\RM{{\cal R}_M}

\numberwithin{equation}{section}

\graphicspath{{images/}}

{\theoremstyle{remark} \newtheorem{remark}{Remark}}
\newtheorem{corollary}{Corollary}

\title{Globally Hyperbolic Regularization of Grad's Moment System in
  One Dimensional Space} 

\author{Zhenning Cai\thanks{School of Mathematical Sciences, Peking
    University, Beijing, China, email: {\tt caizn@pku.edu.cn}.},~~
  Yuwei Fan\thanks{School of Mathematical Sciences, Peking University,
    Beijing, China, email: {\tt ywfan@pku.edu.cn}.},~~ Ruo
  Li\thanks{CAPT, LMAM \& School of Mathematical Sciences, Peking
    University, Beijing, China, email: {\tt rli@math.pku.edu.cn}.}}

\begin{document}
\maketitle
\begin{abstract}
  In this paper, we present a regularization to 1D Grad's moment
  system to achieve global hyperbolicity. The regularization is based
  on the observation that the characteristic polynomial of the
  Jacobian of the flux in Grad's moment system is independent of the
  intermediate coefficients in the Hermite expansion. The method is
  not relied on the form of the collision at all, thus this
  regularization is applicable to the system without collision terms.
  Moreover, the proposed approach is proved to be the unique one if
  only the last moment equation is allowed to be alternated to match
  the condition that the characteristic speeds coincide with the
  Gauss-Hermite interpolation points. The hyperbolic structure of the
  regularized system, including the signal speeds, Riemann invariants
  and the properties of the characteristic waves including the
  rarefaction wave, contact discontinuity and shock are provided in
  the perfect formations.

\vspace*{4mm}
\end{abstract}

\section{Introduction}
Nowadays, the kinetic gas theory is drawing increasing attentions in
the high-tech fields. The kinetic theory is considered as a mesoscopic
description of fluids, which is based on the classical Boltzmann
equation from statistical physics. However, a full accurate mesoscopic
model is still too complex for lots of problems. During a long period
of time, people have been looking for a median model between the
classical macroscopic equations and the Boltzmann equation. This can
be tracked back to the work of Burnett \cite{Burnett}. As is well
known, the Burnett equations are later proved to be linearly unstable
by Bobylev \cite{Bobylev}. Another way leading to linearly stable
intermediate models is the moment method proposed by Grad \cite{Grad}.
Since this method was discarded by Grad himself, very few works
contributed to this area in the last century. However, this field is
becoming active in the recent years, since people find that some
traditional difficulties in the moment equations can be ignored by
some regularizations to these models, e.g. \cite{Levermore, Jin,
  Struchtrup2003, Torrilhon2010}.

This paper is focusing on a major accusation against the moment method
--- the lack of global hyperbolicity for Grad's moment system. This
deficiency directly causes blow-ups when the distribution is far away
from the equilibrium state. It has been reported that increasing the
number of moments shows no improvements in the numerical experiments
\cite{NRxx_new}. Levermore's work \cite{Levermore} gave a theoretical
way to the general globally hyperbolic moment equations, while it is
still far from practical use due to the lack of an analytical form of
his model. Later, using the Pearson-Type-IV distribution, Torrilhon
\cite{Torrilhon2010} also proposed a 13-moment system, which is
globally hyperbolic when reduced to the one-dimensional case, but its
generalization to large number moment systems seems to be difficult.
In this work, we concentrate on the simple 1D case and achieve a
globally hyperbolic regularization to Grad's moment system
successfully.

The first essential observation is that the characteristic polynomial
of Jacobian of the flux of a general Grad's moment system has a simple
expression, which only depends on the macroscopic velocity,
temperature, and two other coefficients in the Hermite expansion of
highest orders. This amazing result directly leads to the possibility
of a globally hyperbolic regularization. It is found that these two
coefficients take the eigenvalues away from the real axis, resulting
in the non-hyperbolicity. We discover an elegant modification to the
last equation of the moment system to eliminate the terms involving
these two terms in the characteristic polynomial and obtain a globally
hyperbolic system. This new hyperbolic system has lots of fascinating
properties. All characteristic fields are either genuinely nonlinear
or linearly degenerate. The investigation into the three kinds of
elementary characteristic waves (rarefaction waves, contact
discontinuities, and shock waves) illustrates substantial similarities
with Euler equations. The regularization proposed is very different
from the classical way which tries to give a reasonable recovery of
the truncated moments, which is justified in the view of
characteristic speeds and order of accuracy. The convergence in the
number of moments is illustrated through the numerical study of a
shock tube problem.

The rest of this paper is arranged as follows: in Section
\ref{sec:NRxx_method}, the Boltzmann equation and the moment method
are revised. In Section \ref{sec:Grad}, a detailed investigation on
the hyperbolicity of 1D Grad's moment system is carried out. The
regularization of the 1D Grad's moment system to achieve global
hyperbolicity is derived in Section \ref{sec:hyperbolic}, with
detailed discussion on its properties. A short discussion on the
moment equations with collision terms is put forward in Section
\ref{sec:col}. Section \ref{sec:num} is devoted to the numerical study
of a shock tube problem. Finally, some concluding remarks are given in
Section \ref{sec:conclusion}.


\section{The moment method in kinetic theory} \label{sec:NRxx_method}
In the kinetic gas theory, the state of a gas on the microscopic level
is described by the velocity distribution function on each spatial
point $\bx \in \Omega \subset \bbR^D$. For a time-evolving problem,
the distribution function can be described as
\begin{equation}
F:\bbR^+ \times \Omega \times \bbR^D \rightarrow \bbR^+ \cup \{0\},
  \quad (t, \bx, \bxi) \mapsto F(t, \bx, \bxi),
\end{equation}
where $t$ is the time and $\bxi$ denotes the velocity of microscopic
gas particles. As in \cite{Grad}, we introduce the mass density
\begin{equation}
f(t, \bx, \bxi) = m F(t, \bx, \bxi),
\end{equation}
$m$ being the mass of the molecule. The physical case is $D = 3$,
while in this paper, we only consider a 1D model problem with $D = 1$.
Thus, $\bx$ and $\bxi$ will be written in plain font as $x$ and $\xi$
later on.

\subsection{The Boltzmann equation and conservation laws}
\label{sec:boltzmann_equations}
The mass density $f$ satisfies the Boltzmann equation, which reads
\begin{equation} \label{eq:Boltzmann}
\frac{\partial f}{\partial t} +
  \xi \frac{\partial f}{\partial x} = Q(f,f),
\end{equation}
where $Q(f,f)$ is the collision term with a complex expression, which
models the interaction between particles. In most part of this paper,
we only consider the collisionless case, thus $Q(f,f) = 0$ will be
assumed if not specified. However, the readers may keep in mind that
our final aim is to provide an improved description of the Boltzmann
equation with collision term using the moment method, and we will
return to this topic in Section \ref{sec:col}.

The basic variables, including the density, the momentum density and
total energy density, are defined as
\begin{equation}
\begin{aligned}
\rho(t,x) &= \int_{\bbR} f(t, x, \xi) \dd \xi,\\
\rho(t,x) u(t,x) &= \int_{\bbR}\xi f(t,x,\xi)\dd\xi,\\
\frac{1}{2} \rho(t,x) |u(t,x)|^2 + \frac{1}{2} \rho(t,x)\theta(t,x)
  &= \int_{\bbR} \frac{1}{2}|\xi|^2 f(t,x,\xi)\dd\xi.
\end{aligned}
\end{equation}
Here $u$ is the macroscopic velocity, and $\theta$ is the
multiplication of gas constant and temperature. Multiplying the
Boltzmann equation \eqref{eq:Boltzmann} by $(1, \xi, \xi^2/2)^T$,
integrating both sides over $\bbR$ with respect to $\xi$, and then
making some simplifications, we get the non-conservative form of the
conservation laws as
\begin{subequations}\label{eq:conservation_laws}
\begin{align}
&\frac{\partial \rho}{\partial t} + u \frac{\partial \rho}{\partial x}
  + \rho \frac{\partial u}{\partial x} = 0, \\
& \rho \frac{\partial u}{\partial t} + \frac{\partial p}{\partial x}
  + \rho u \frac{\partial u}{\partial x} = 0, \\
&\frac{1}{2}\rho \frac{\partial \theta}{\partial t}
  + \frac{1}{2} \rho u \frac{\partial \theta}{\partial x}
  + \frac{\partial q}{\partial x} + p \frac{\partial u}{\partial x}
  = 0,
\end{align}
\end{subequations}
where $p$ is the pressure and $q$ is the heat flux. They are defined
as
\begin{equation}
p = \rho \theta, \quad
  q = \frac{1}{2} \int_{\bbR} (\xi - u)^3 f \dd \xi.
\end{equation}

\subsection{The moment method} \label{sec:moment_method_detail}
The moment method was raised by Grad in \cite{Grad}, where a
thirteen moment system was introduced. However, systems with large
moment numbers are not investigated until recently (e.g. \cite{Au,
Weiss, NRxx, NRxx_new}). Here we use the notations in \cite{NRxx,
NRxx_new}, and expand the $f(t, x, \xi)$ as
\begin{equation} \label{eq:expansion}
f(t,x,\xi) = \sum_{k \in \bbN}f_k(t,x)
  \mathcal{H}_{\theta(t,x),k}
  \left( \frac{\xi - u(t,x)}{\sqrt{\theta(t,x)}} \right),
\end{equation}
where
\begin{equation}
\mathcal{H}_{\theta,k}(v) =
  \frac{1}{\sqrt{2\pi}} \theta^{-\frac{k+1}{2}}
  \He_k(v) \exp \left( -\frac{v^2}{2} \right),
\end{equation}
where $\He_k$ is the $k$-th Hermite polynomial, defined by
\begin{equation}
\He_k(x) = (-1)^k \exp \left( \frac{x^2}{2} \right)
  \frac{\dd^k}{\dd x^k} \exp \left( -\frac{x^2}{2} \right).
\end{equation}
Based on this expansion, some simple properties can be deduced:
\begin{equation}
f_0 = \rho, \quad f_1 = f_2 = 0, \quad q = 3f_3.
\end{equation}
If we put \eqref{eq:expansion} into the Boltzmann equation
\eqref{eq:Boltzmann}, the equation for each moment can be deduced as
\begin{equation} \label{eq:discrete_f}
\begin{split}
\frac{\partial f_k}{\partial t} & 
  - f_{k-1} \frac{\theta}{\rho} \frac{\partial \rho}{\partial x}
  + (k+1) f_k \frac{\partial u}{\partial x} + \left(
    \frac{1}{2}\theta f_{k-3} +
    \frac{k-1}{2} f_{k-1}
  \right) \frac{\partial \theta}{\partial x} \\
& - \frac{3}{\rho} f_{k-2} \frac{\partial f_3}{\partial x}
  + \theta \frac{\partial f_{k-1}}{\partial x}
  + u \frac{\partial f_k}{\partial x}
  + (k+1) \frac{\partial f_{k+1}}{\partial x} = 0,
	\quad \text{for }k \geqslant 3.
\end{split}
\end{equation}
For details, we refer the readers to \cite{NRxx_new}. The conservation
laws \eqref{eq:conservation_laws} together with \eqref{eq:discrete_f}
form a moment system with infinite number of equations. In order to
get a closed system with finite number of equations, one can follow
Grad's idea \cite{Grad} and let $f_{M+1} = 0$ for some $M \geqslant
3$. Thus a closed system with $M+1$ moments is obtained.


\section{Hyperbolicity of Grad's moment systems} \label{sec:Grad}
A 1D quasilinear system
\begin{equation}
\frac{\partial \bq}{\partial t} +
  {\bf A}(\bq) \frac{\partial \bq}{\partial x} = 0
\end{equation}
is pronounced to be hyperbolic for a particular $\bq_0$ if the matrix ${\bf
A}(\bq_0)$ is diagonalizable with real eigenvalues. For Grad's
systems, the hyperbolicity can only be obtained in the vicinity of
Maxwellian \cite{Muller, Brini, Torrilhon2010}. The loss of
hyperbolicity makes the Grad's system overdetermined for strongly
non-equilibrium gases, and severely restrict the application of moment
methods. In this section, we are going to study the 1D model problem
and find the way in which high order moments affect the hyperbolicity
of moment system.

Let $\bw_M = (\rho, u, \theta, f_3,\cdots,f_M)^T \in \bbR^{M+1}$, $M
\in \bbN\text{ and }M \geqslant 2$. The Grad's moment system
\eqref{eq:conservation_laws} and \eqref{eq:discrete_f} with $f_{M+1} =
0$ is then written as
\begin{equation} \label{eq:Grad}
\frac{\partial \bw_M}{\partial t} +
  {\bf A}_M \frac{\partial \bw_M}{\partial x} = 0,
\end{equation}
where ${\bf A}_M$ is a lower Hessenberg matrix defined in
\eqref{eq:A_M}. We write the matrix in a simplified formation with a
translation and similarity transformation. Let us define
\begin{equation}
{\bf \Lambda} = \mathrm{diag} \left\{
  1, \rho \theta^{- 1 / 2}, \frac{1}{2} \rho \theta^{- 1},
  \theta^{- 3 / 2}, \cdots, \theta^{- M / 2}
\right\}, \qquad
  g_j = \frac{f_j}{\rho \theta^{j / 2}}, \quad j = 3, \cdots, M.
\end{equation}
Then 
\begin{equation}
{\bf A}_M = u{\bf I}+ \sqrt{\theta} {\bf \Lambda}^{-1}
   \tilde{\bf A}_M {\bf \Lambda},
\end{equation}
where $\tilde{\bf A}_M$ is defined in \eqref{eq:tilde_A_M}. Thus, if
\begin{displaymath}
\tilde{\lambda}_j, \quad j=1,\cdots,M+1
\end{displaymath}
are all the eigenvalues of $\tilde{\bf A}_M$, then
\begin{displaymath}
u + \tilde{\lambda}_j \sqrt{\theta}, \quad  j=1,\cdots,M+1
\end{displaymath}
are all the eigenvalues of ${\bf A}_M$.

\begin{landscape}
\renewcommand{\arraystretch}{1.5}
\begin{equation} \label{eq:A_M}
{\bf A}_M = \begin{pmatrix}
  u & \rho & 0 & \hdotsfor{6} & 0\\
  \theta / \rho & u & 1 & 0 & \hdotsfor{5} & 0\\
  0 & 2 \theta & u & 6 / \rho & 0 & \hdotsfor{4} & 0\\
  0 & 4 f_3 & \rho \theta / 2 & u & 4 & 0 & \hdotsfor{3} & 0\\
  - \theta f_3 / \rho & 5 f_4 & 3 f_3 / 2 & \theta & u & 5 & 0 & \hdotsfor{2} & 0\\
  \hdotsfor{10} \\
  -\theta f_{M - 2} / \rho & M f_{M - 1} &
    \frac{1}{2}[(M - 2) f_{M - 2} + \theta f_{M - 4}] &
    -3 f_{M - 3} / \rho & 0 & \cdots & 0 &\theta & u & M\\
  -\theta f_{M-1} / \rho & (M+1) f_M &
    \frac{1}{2}[(M-1) f_{M - 1} + \theta f_{M - 3}] &
    -3 f_{M-2} / \rho & 0 & \hdotsfor{2} & 0 & \theta & u
\end{pmatrix}
\end{equation}

\vspace{1cm}

\begin{equation} \label{eq:tilde_A_M}
\tilde{\bf A}_M = \begin{pmatrix}
  0 & 1 & 0 & \hdotsfor{6} & 0\\
  1 & 0 & 2 & 0 & \hdotsfor{5} & 0\\
  0 & 1 & 0 & 3 & 0 & \hdotsfor{4} & 0\\
  0 & 4 g_3 & 1 & 0 & 4 & 0 & \hdotsfor{3} & 0\\
  -g_3 & 5 g_4 & 3 g_3 & 1 & 0 & 5 & 0 & \hdotsfor{2} & 0 \\
  \hdotsfor{10} \\
  -g_{M - 2} & M g_{M - 1} &
    (M - 2) g_{M - 2} + g_{M - 4} &
    -3 g_{M - 3} & 0 & \cdots & 0 & 1 & 0 & M\\
  -g_{M-1} & (M+1) g_M &
    (M-1) g_{M - 1} + g_{M - 3} &
    -3 g_{M-2} & 0 & \hdotsfor{2} & 0 & 1 & 0
\end{pmatrix}
\end{equation}
\end{landscape}

The matrix $\tilde{\bf A}_M$ can be considered as ``simple'' in a
sense. It contains only dimensionless variables $g_3, \cdots, g_M$
with linear dependence. The diagonal elements of $\tilde{\bf A}_M$ are
all vanished, and the subdiagonal entries are all $1$. The
superdiagonal elements are equal to their row numbers. Meanwhile,
apart from the tridiagonal part, only the first four columns are
nonzero. These formation give us possibility to study its eigenvalues.

We first present the main result of this section in Theorem
\ref{thm:cp_Grad}. In this paper, $|\cdot|$ is used to denote the
determinant of a matrix.
\begin{theorem} \label{thm:cp_Grad}
The characteristic polynomial of $\tilde{\bf A}_M$ is
\begin{equation}
\left| \lambda {\bf I} - \tilde{\bf A}_M \right| =
  \He_{M+1}(\lambda) - \frac{1}{2} (M + 1)! \cdot
    [(\lambda^2 - 1) g_{M-1} + 2\lambda g_M].
\end{equation}
\end{theorem}
The result is incredibly simple, and therefore gives us a realistic
possibility to make some kind of regularization to gain global
hyperbolicity, which will be discussed in the next section. To proof
this theorem, we need the following two lemmas.

\begin{lemma} \label{lem:det_diff}
Suppose that a square matrix ${\bf A} = (a_{ij})$ depends on $N$
variables $x_1, \cdots, x_N$. Then the partial derivatives of $|\bf
A|$ can be calculated as
\begin{equation} \label{eq:det_diff}
\frac{\partial |{\bf A}|}{\partial x_k} =
  \sum_{i,j} (-1)^{i+j} \frac{\partial a_{ij}}{\partial x_k} A^{ij},
  k = 1, \cdots, N.
\end{equation}
where $A^{ij}$ is the $(i,j)$-th minor of matrix $\bf A$, which is
defined to be the determinant of the submatrix obtained by removing
from $\bf A$ its $i$-th row and $j$-th column.
\end{lemma}
This is a familiar result in linear algebra, and will not be proved
here.

\begin{lemma} \label{lem:recur}
Define tridiagonal matrices
\begin{equation} \label{eq:D}
\renewcommand{\arraystretch}{1.2}
{\bf D}_j = \begin{pmatrix}
\lambda & -(j+1) & 0 & \hdotsfor{3} & 0 \\
-1 & \lambda & -(j+2) & 0 & \hdotsfor{2} & 0 \\
0 & -1 & \lambda & -(j+3) & 0 & \cdots & 0 \\
\hdotsfor{7}\\
0 & \hdotsfor{2} & 0 & -1 & \lambda & -M \\
0 & \hdotsfor{3} & 0 & -1 & \lambda \\
\end{pmatrix}, \quad 0 \leqslant j \leqslant M.
\end{equation}
The following relations for the determinants of ${\bf D}_j$ hold:
\begin{equation}
|{\bf D}_j| = \lambda |{\bf D}_{j+1}| - (j+1) |{\bf D}_{j+2}|,
  \quad 0 \leqslant j \leqslant M - 2.
\end{equation}
\end{lemma}

\begin{proof}
For $0 \leqslant j \leqslant M - 1$, $D_j$ can be written as
\begin{equation} \label{eq:recur}
{\bf D}_j = \begin{pmatrix}
  \lambda & -(j+1) \be_1^T \\
  -\be_1 & {\bf D}_{j+1}
\end{pmatrix},
\end{equation}
where $\be_1$ is the unit vector $(1, 0, \cdots, 0)^T$. When $\lambda
\neq 0$, since
\begin{equation} \label{eq:mat_prod}
\begin{pmatrix}
{\bf I} & {\bf 0} \\ \lambda^{-1} \be_1 & {\bf I}
\end{pmatrix}
\begin{pmatrix}
\lambda & -(j+1)\be_1^T \\ -\be_1 & {\bf D}_{j+1}
\end{pmatrix} = 
\begin{pmatrix}
\lambda & -(j+1)\be_1^T \\
{\bf 0} & {\bf D}_{j+1} - (j+1) \lambda^{-1} \be_1 \be_1^T
\end{pmatrix},
\end{equation}
the equality
\begin{equation}
|{\bf D}_j| = \lambda \left|
  {\bf D}_{j+1} - (j+1) \lambda^{-1} \be_1 \be_1^T
\right|
\end{equation}
is obtained by taking determinants on both sides of
\eqref{eq:mat_prod}. When $0 \leqslant j \leqslant M - 2$, we use
\eqref{eq:recur} again and get
\begin{equation}
\begin{split}
|{\bf D}_j| & = \lambda \left|
  {\bf D}_{j+1} - (j+1) \lambda^{-1} \be_1 \be_1^T
\right| = \lambda \begin{vmatrix}
  \lambda - (j+1) \lambda^{-1} & -(j+2) \be_1^T \\
  -\be_1 & {\bf D}_{j+2}
\end{vmatrix} \\
&= \lambda \left(
  \begin{vmatrix}
    \lambda & -(j+2)\be_1^T \\ -\be_1 & {\bf D}_{j+2}
  \end{vmatrix} +
  \begin{vmatrix}
    -(j+1) \lambda^{-1} & {\bf 0} \\ -\be_1 & {\bf D}_{j+2}
  \end{vmatrix}
\right) = \lambda |{\bf D}_{j+1}| - (j+1) |{\bf D}_{j+2}|.
\end{split}
\end{equation}
If $\lambda = 0$, the continuity of $|{\bf D}_j|$ with respect to
$\lambda$ gives the same result.
\end{proof}

Now we prove Theorem \ref{thm:cp_Grad}.
{
\renewcommand\proofname{Proof of Theorem \ref{thm:cp_Grad}}
\begin{proof}
We start the proof by calculating $\partial |\lambda {\bf I} -
\tilde{\bf A}_M| / \partial g_j$ for $3 \leqslant j \leqslant M - 3$.
From \eqref{eq:tilde_A_M}, one may find that $g_j$ only appears in five
entries of the matrix. Their positions are
\begin{displaymath}
(j+2,1), \quad (j+1,2), \quad (j+2,3), \quad (j+4,3), \quad (j+3,4),
\end{displaymath}
which are illustrated in Figure \ref{fig:entries}. Thus, according to
Lemma \ref{lem:det_diff}, only five terms appear in the right hand
side of \eqref{eq:det_diff}. Now we will consider them one by one.
Below we denote $\lambda {\bf I} - \tilde{\bf A}_M = (c_{ij})$, and
use $C^{i,j}$ to denote the $(i,j)$-th minor of $\lambda {\bf I} -
\tilde{\bf A}_M$.
\begin{enumerate}
\item As in Figure \ref{fig:minor_j+2_1}, $C^{j+2,1}$ is
presented as the product of the determinants of two matrices. One is a
lower triangular matrix whose diagonal elements are
$-1,\cdots,-(j+1)$, and the other is a lower right block of $\lambda
{\bf I} - \tilde{\bf A}_M$, which is actually ${\bf D}_{j+2}$ defined
in \eqref{eq:D}. Therefore, we obtain
\begin{equation}
C^{j+2,1} = (-1)^{j+1} (j+1)! \cdot |{\bf D}_{j+2}|.
\end{equation}
Since $c_{j+2,1} = g_j$, one has
\begin{equation} \label{eq:term1}
(-1)^{j+2+1} \frac{\partial c_{j+2,1}}{\partial g_j} C^{j+2,1}
  = (-1)^{j+1} \cdot 1 \cdot (-1)^{j+1} (j+1)! \cdot |{\bf D}_{j+2}|
  = (j+1)! \cdot |{\bf D}_{j+2}|.
\end{equation}
\item Figure \ref{fig:minor_j+1_2} shows that $C^{j+1,2}$ is
factorized into three parts: the first part is $\lambda$, the second
is a lower triangular matrix with diagonal elements $-2,\cdots,-j$,
and the third one is ${\bf D}_{j+1}$. Since $c_{j+1,2} = -(j+1)g_j$,
we get
\begin{equation} \label{eq:term2}
\begin{split}
(-1)^{j+1+2} \frac{\partial c_{j+1,2}}{\partial g_j} C^{j+1,2}
  & = (-1)^{j+1} \cdot [-(j+1)] \cdot
    (-1)^{j-1} j! \lambda \cdot |{\bf D}_{j+1}| \\
  & = -(j+1)! \cdot \lambda |{\bf D}_{j+1}|.
\end{split}
\end{equation}
\item $C^{j+2,3}$ is illustrated in Figure \ref{fig:minor_j+2_3}, from
which one finds $C^{j+2,3}$ is the product of the determinants of
three matrices. The first matrix is a $2\times 2$ upper left block of
$\lambda {\bf I} - \tilde{\bf A}_M$, for which we have
\begin{equation}
\begin{vmatrix}
\lambda & -1 \\ -1 & \lambda
\end{vmatrix} = \lambda^2 - 1.
\end{equation}
And the other two blocks are similar as the last case. Using
$c_{j+2,3} = -j g_j$, we have
\begin{equation} \label{eq:term3}
\begin{split}
(-1)^{j+2+3} \frac{\partial c_{j+2,3}}{\partial g_j} C^{j+2,3}
  & = (-1)^{j+1} \cdot (-j) \cdot
    (-1)^{j-1} \frac{1}{2} (j+1)! (\lambda^2-1) \cdot |{\bf D}_{j+2}| \\
  & = -\frac{j}{2} (j+1)! \cdot (\lambda^2 - 1) |{\bf D}_{j+2}|.
\end{split}
\end{equation}
\item The structure of $C^{j+4,3}$ is plotted in Figure
\ref{fig:minor_j+4_3}, which is very similar as $C^{j+2,3}$. Therefore
we directly write the result:
\begin{equation} \label{eq:term4}
\begin{split}
(-1)^{j+4+3} \frac{\partial c_{j+4,3}}{\partial g_j} C^{j+4,3}
  & = (-1)^{j+1} \cdot (-1) \cdot
    (-1)^{j+1} \frac{1}{2} (j+3)! (\lambda^2-1) \cdot |{\bf D}_{j+4}| \\
  & = -\frac{1}{2} (j+3)! \cdot (\lambda^2 - 1) |{\bf D}_{j+4}|,
\end{split}
\end{equation}
where we have used $c_{j+4,3} = -(j+2) g_{j+2} - g_j$. Note that we
define $|{\bf D}_{M+1}| = 1$ in order that \eqref{eq:term4} is correct
for $j = M - 3$.
\item Similar as $C^{j+2,3}$ and $C^{j+4,3}$, the minor $C^{j+3,4}$
is also factorized into the determinants of three matrices as in
Figure \ref{fig:minor_j+3_4}, while the first matrix is the $3\times
3$ upper left block of $\lambda {\bf I} - \tilde{\bf A}_M$, whose
determinant is
\begin{equation}
\begin{vmatrix}
\lambda & -1 & 0 \\ -1 & \lambda & -2 \\ 0 & -1 & \lambda
\end{vmatrix} = \lambda^3 - 3 \lambda.
\end{equation}
Thus the last term becomes
\begin{equation} \label{eq:term5}
\begin{split}
(-1)^{j+3+4} \frac{\partial c_{j+3,4}}{\partial g_j} C^{j+3,4}
  & = (-1)^{j+1} \cdot 3 \cdot (-1)^{j-1} \frac{1}{6} (j+2)!
    (\lambda^3-3 \lambda) \cdot |{\bf D}_{j+3}| \\
  & = \frac{1}{2} (j+2)! \cdot (\lambda^3 - 3 \lambda) |{\bf D}_{j+3}|.
\end{split}
\end{equation}
\end{enumerate}
Collecting \eqref{eq:term1}, \eqref{eq:term2}, \eqref{eq:term3},
\eqref{eq:term4} and \eqref{eq:term5}, we finally get
\begin{equation} \label{eq:cp_diff}
\begin{split}
\frac{\partial |\lambda {\bf I} - \tilde{\bf A}_M|}{\partial g_j} &=
  (j+1)! \cdot \bigg[ |{\bf D}_{j+2}| - \lambda |{\bf D}_{j+1}| -
    \frac{j}{2} (\lambda^2 - 1) |{\bf D}_{j+2}| \\
&\qquad -\frac{(j+3)(j+2)}{2} (\lambda^2 - 1) |{\bf D}_{j+4}| +
    \frac{j+2}{2} (\lambda^3 - 3 \lambda) |{\bf D}_{j+3}| \bigg].
\end{split}
\end{equation}
This expression will be further simplified using Lemma
\ref{lem:recur}.  Since \eqref{eq:recur} also holds for $j = M-1$ if
we define $|{\bf D}_{M+1}| = 1$, the following relation is deduced:
\begin{equation}
\begin{split}
& |{\bf D}_{j+2}| - \lambda |{\bf D}_{j+1}| -
  \frac{j}{2} (\lambda^2 - 1) |{\bf D}_{j+2}| \\
={} & |{\bf D}_{j+2}| -
  \lambda (\lambda |{\bf D}_{j+2}| - (j+2) |{\bf D}_{j+3}|) -
  \frac{j}{2} (\lambda^2 - 1) |{\bf D}_{j+2}| \\
={} & - \frac{j+2}{2} (\lambda^2 - 1) |{\bf D}_{j+2}|
  + (j+2)\lambda |{\bf D}_{j+3}| \\
={} & - \frac{j+2}{2} (\lambda^2 - 1)
    [\lambda |{\bf D}_{j+3}| - (j+3) |{\bf D}_{j+4}|]
  + (j+2)\lambda |{\bf D}_{j+3}| \\
={} & - \frac{j+2}{2} (\lambda^3 - 3\lambda) |{\bf D}_{j+3}|
  + \frac{(j+2)(j+3)}{2} (\lambda^2 - 1) |{\bf D}_{j+4}|.
\end{split}
\end{equation}
Substituting this equation into \eqref{eq:cp_diff}, we conclude
\begin{equation}
\frac{\partial |\lambda {\bf I} - \tilde{\bf A}_M|}{\partial g_j} = 0,
  \quad 3\leqslant j \leqslant M - 3.
\end{equation}
It is clear that $g_3, \cdots, g_{M-3}$ do not appear in the
characteristic polynomial of $\tilde{\bf A}_M$.

For $j = M-2, M-1, M$, the entries containing $g_j$ still locate in
the matrix as Figure \ref{fig:entries}, while some items are missing
due to the cut-off. Therefore, if we define $|{\bf D}_j| = 0$ for $j >
M+1$, then \eqref{eq:cp_diff} still applies for $j = M-2, M-1, M$.
Note that such definition leads to
\begin{equation}
|{\bf D}_M| = \lambda |{\bf D}_{M+1}| - (M+1) |{\bf D}_{M+2}|,
\end{equation}
therefore $g_{M-2}$ does not appear in $|\lambda {\bf I} - \tilde{\bf
A}_M|$ either. Moreover, we have
\begin{gather}
\label{eq:partial_g_M-1}
\begin{split}
\frac{\partial |\lambda {\bf I} - \tilde{\bf A}_M|}{\partial g_{M-1}}
&= M! \cdot \left[ |{\bf D}_{M+1}| - \lambda |{\bf D}_M|
  - \frac{M-1}{2} (\lambda^2 - 1) |{\bf D}_{j+2}| \right] \\
&= -\frac{(M+1)!}{2} (\lambda^2 - 1),
\end{split} \\
\label{eq:partial_g_M}
\frac{\partial |\lambda {\bf I} - \tilde{\bf A}_M|}{\partial g_{M}}
= (M+1)! \cdot (-\lambda |{\bf D}_{M+1}|) = -(M+1)! \cdot \lambda.
\end{gather}
Since \eqref{eq:partial_g_M-1} and \eqref{eq:partial_g_M} hold for any
$g_j$, $3 \leq j \leq M$, we write the characteristic polynomial of
$\tilde{\bf A}_M$ as
\begin{equation}
|\lambda {\bf I} - \tilde{\bf A}_M| = C(\lambda) -
  \frac{(M+1)!}{2} [ (\lambda^2 - 1) g_{M-1} + 2\lambda g_M],
\end{equation}
where $C(\lambda)$ is a function of $\lambda$.

Now it only remains to determine $C(\lambda)$, which is done by
assigning $g_3, \cdots, g_M$ to be zero, and then calculating the
characteristic polynomial of $\tilde{\bf A}_M$. In this case, it is
easy to find
\begin{equation}
|\lambda {\bf I} - \tilde{\bf A}_M| = C(\lambda) = |{\bf D}_0|,
  \qquad \text{if } g_3 = \cdots = g_M = 0.
\end{equation}
Meanwhile, the following relation between $|{\bf D}_j|$ and Hermite
polynomials is discovered:
\begin{equation}
\begin{gathered}
|{\bf D}_{M+1}| = \He_0(\lambda) = 1, \qquad
|{\bf D}_M| = \He_1(\lambda) = \lambda, \\
|{\bf D}_{j}| = \lambda |{\bf D}_{j+1}| - (j+1) |{\bf D}_{j+2}|, \\
\He_j(\lambda) = \lambda \He_{j-1}(\lambda) - (j-1) \He_{j-2}(\lambda).
\end{gathered}
\end{equation}
This reveals that
\begin{equation}
|{\bf D}_j| = \He_{M+1-j}(\lambda), \quad 0 \leqslant j \leqslant M+1.
\end{equation}
Hence $C(\lambda) = |{\bf D}_0| = \He_{M+1}(\lambda)$. This completes
the proof of Theorem \ref{thm:cp_Grad}.
\end{proof}
}

\begin{figure}[p]
\centering
\subfigure[Entries containing $g_j$]{
  \label{fig:entries}
  \includegraphics[width=.4\textwidth]{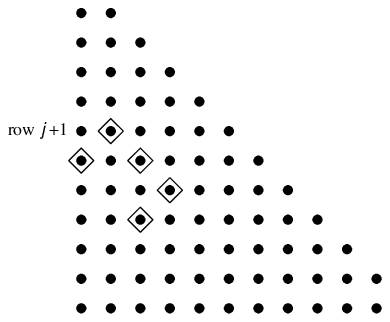}
}
\hspace{40pt}
\subfigure[The $(j+2,1)$-th minor of $\lambda {\bf I} - \tilde{\bf A}_M$]{
  \label{fig:minor_j+2_1}
  \includegraphics[width=.4\textwidth]{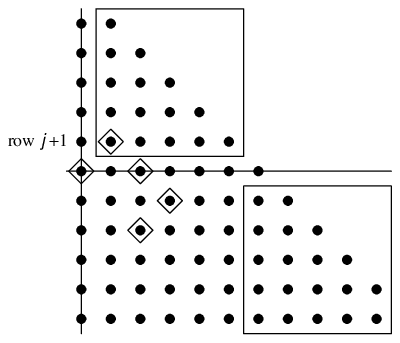}
}\\[20pt]
\subfigure[The $(j+1,2)$-th minor of $\lambda {\bf I} - \tilde{\bf A}_M$]{
  \label{fig:minor_j+1_2}
  \includegraphics[width=.4\textwidth]{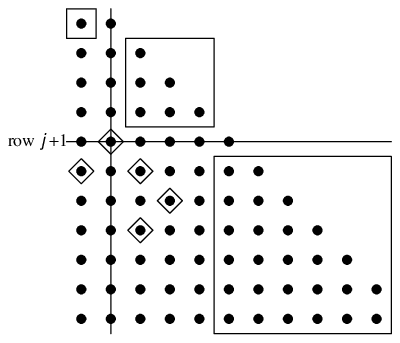}
}
\hspace{40pt}
\subfigure[The $(j+2,3)$-th minor of $\lambda {\bf I} - \tilde{\bf A}_M$]{
  \label{fig:minor_j+2_3}
  \includegraphics[width=.4\textwidth]{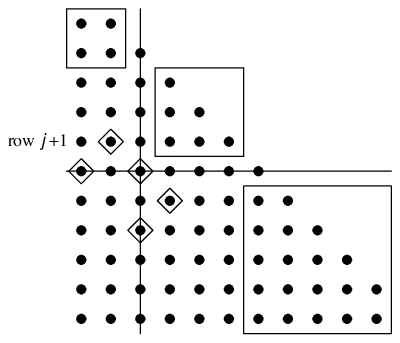}
}\\[20pt]
\subfigure[The $(j+4,3)$-th minor of $\lambda {\bf I} - \tilde{\bf A}_M$]{
  \label{fig:minor_j+4_3}
  \includegraphics[width=.4\textwidth]{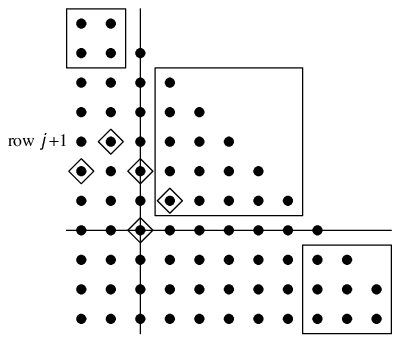}
}
\hspace{40pt}
\subfigure[The $(j+3,4)$-th minor of $\lambda {\bf I} - \tilde{\bf A}_M$]{
  \label{fig:minor_j+3_4}
  \includegraphics[width=.4\textwidth]{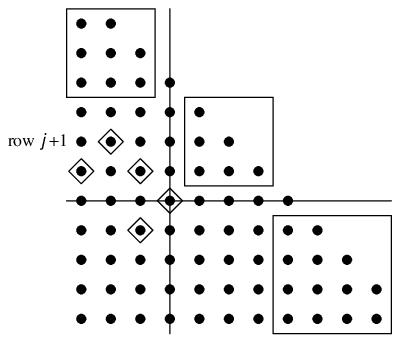}
}
\caption{The Hessenberg matrix $\tilde{A}_M$ and its minors of the
elements containing $g_j$.}
\end{figure}

Theorem \ref{thm:cp_Grad} reveals that the hyperbolicity can only be
obtained in a particular region $(g_{M-1}, g_M) \in \Omega_M$ for the
$(M+1)$-moment sytem. Since the roots of Hermite polynomials are all
real, the origin must lie in $\Omega_M$. The region $\Omega_M$ for
$M=4$ to $9$ are plotted in Figure \ref{fig:hyp_reg}, among which the
result for $M = 4$ has been obtained in \cite{Torrilhon2010}, agreeing
with ours with proper scaling and translation.

\begin{figure}[p] 
\centering
\subfigure[$M = 4$]{
   \includegraphics[scale=.44]{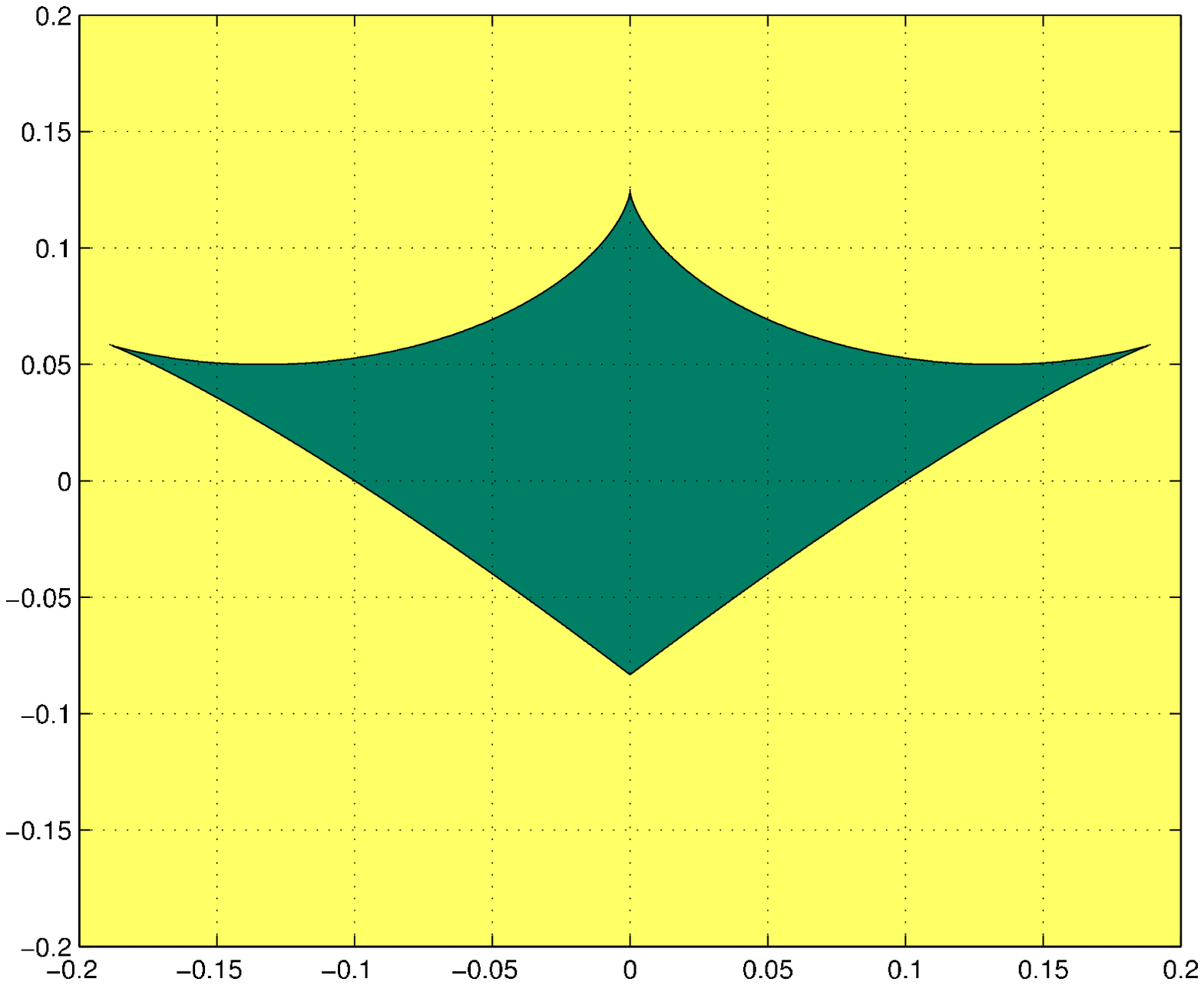}
}
\subfigure[$M = 5$]{
   \includegraphics[scale=.44]{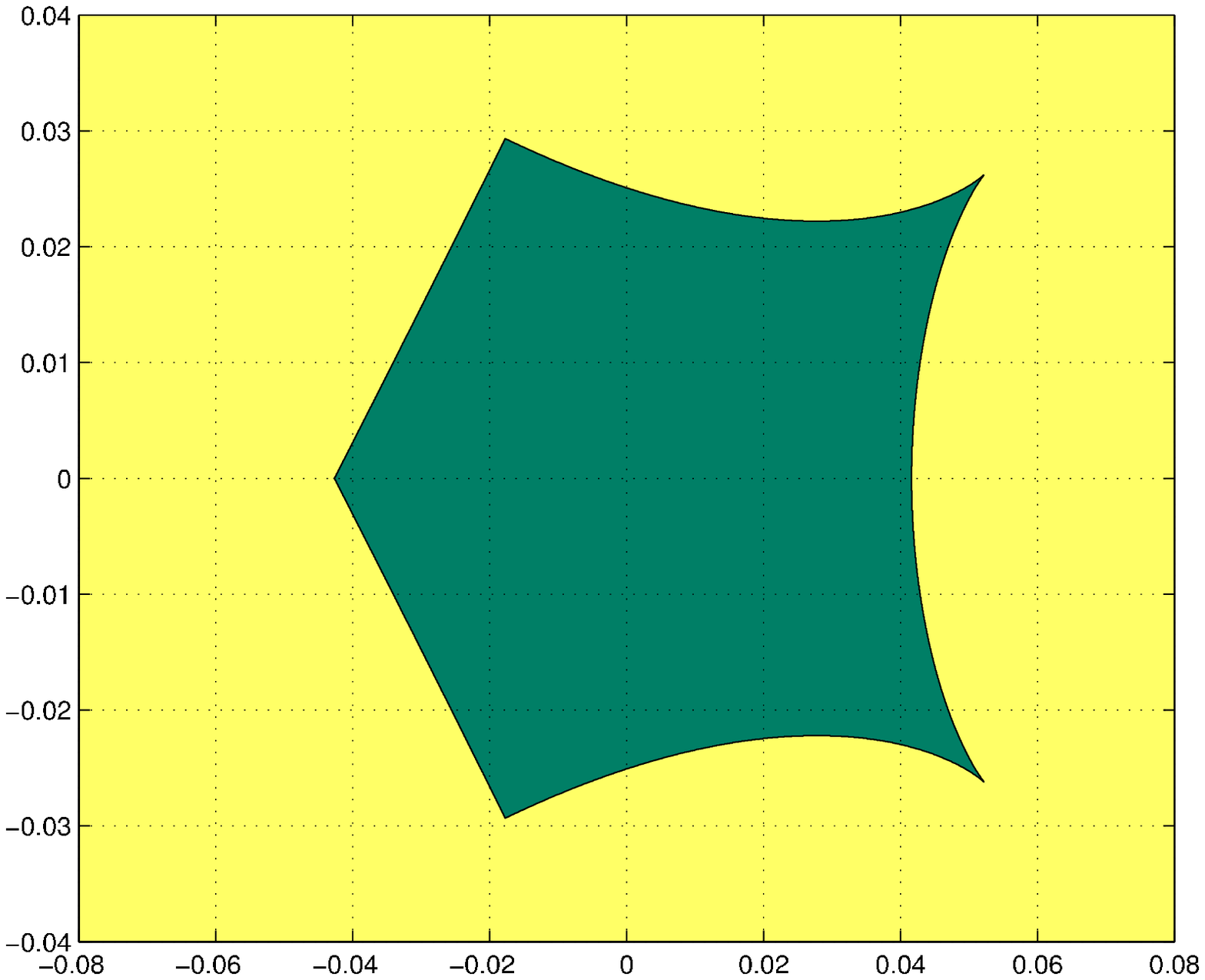}
}\\
\subfigure[$M = 6$]{
   \includegraphics[scale=.44]{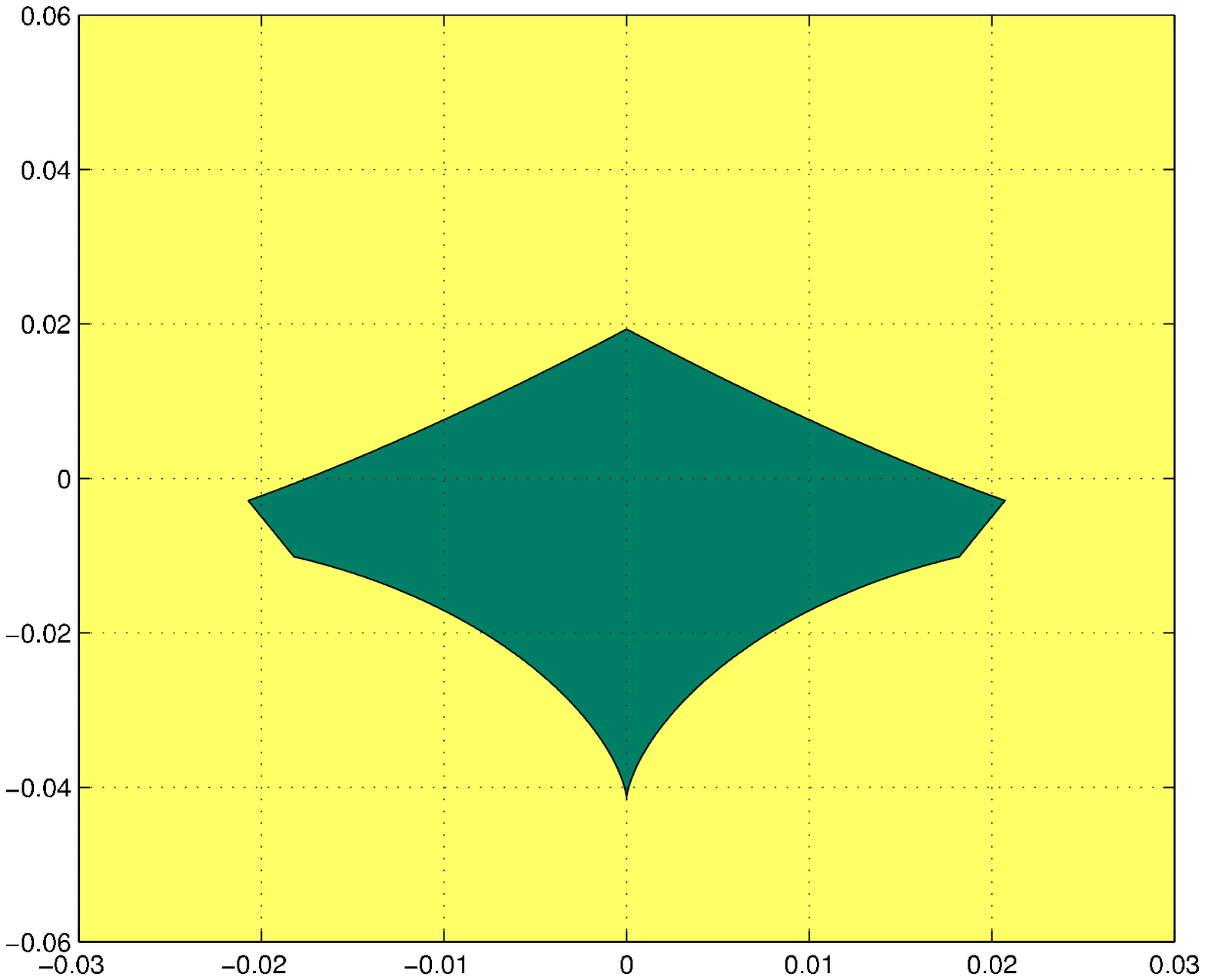}
}
\subfigure[$M = 7$]{
   \includegraphics[scale=.44]{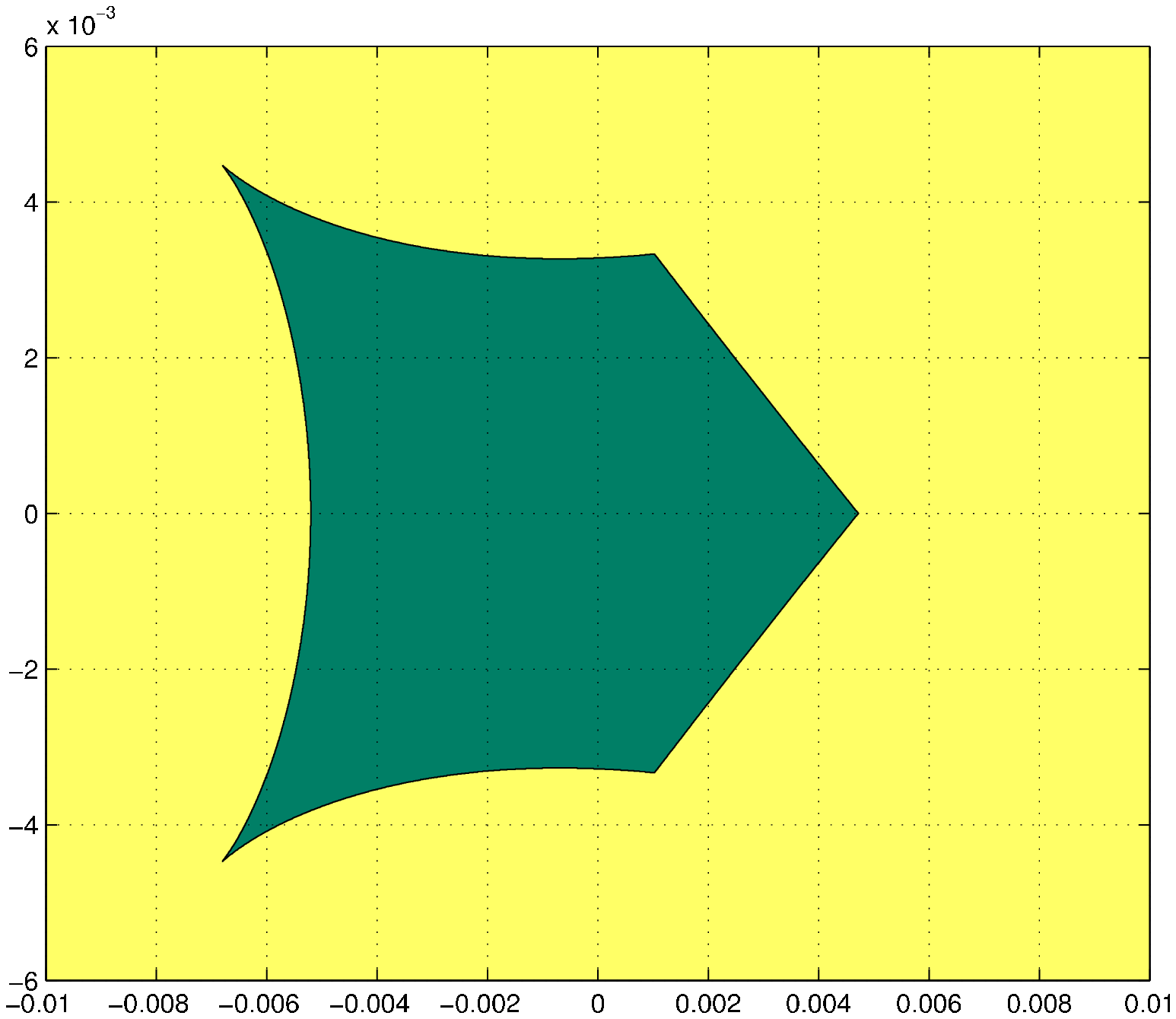}
}\\
\subfigure[$M = 8$]{
   \includegraphics[scale=.44]{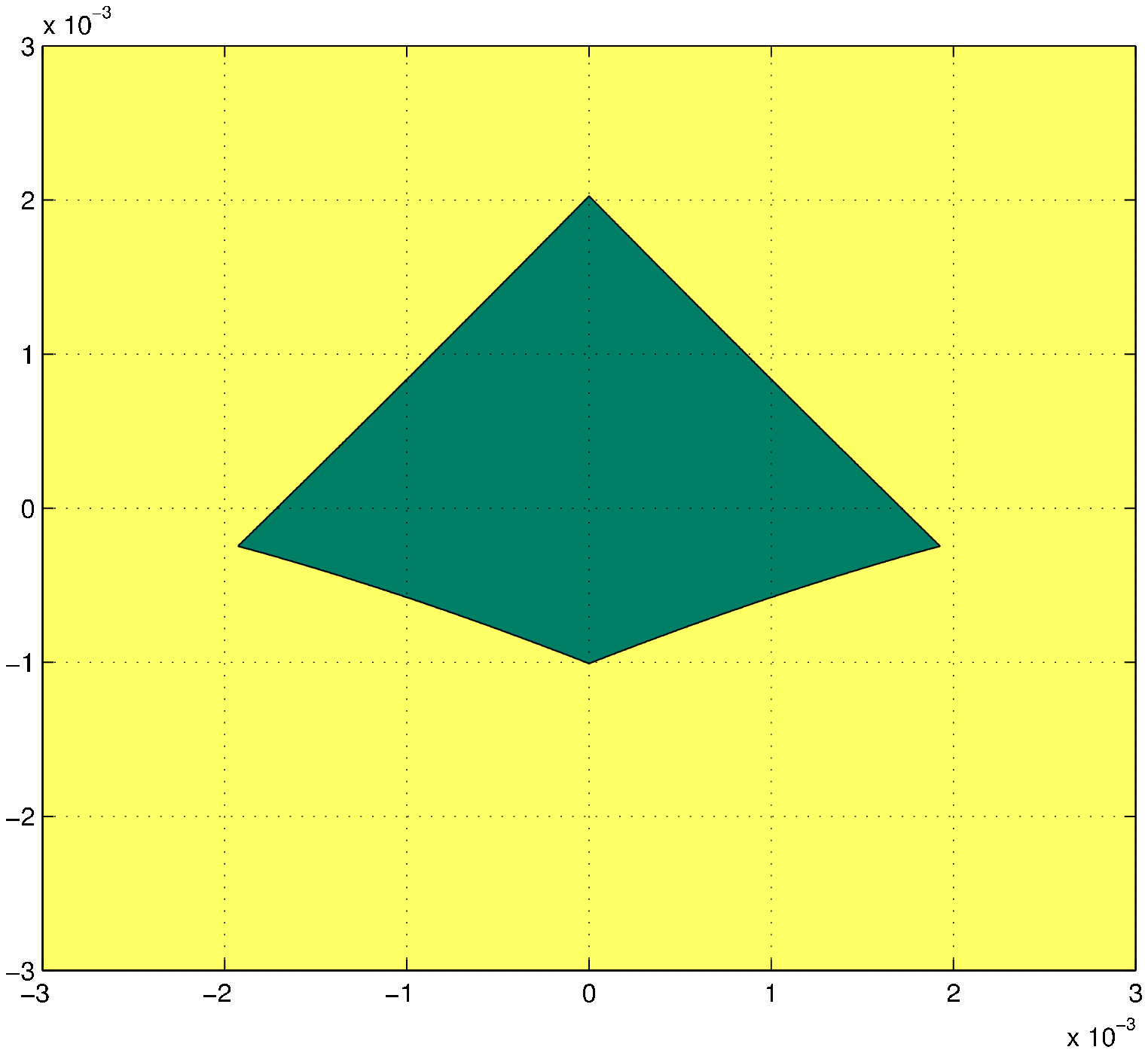}
}
\subfigure[$M = 9$]{
   \includegraphics[scale=.44]{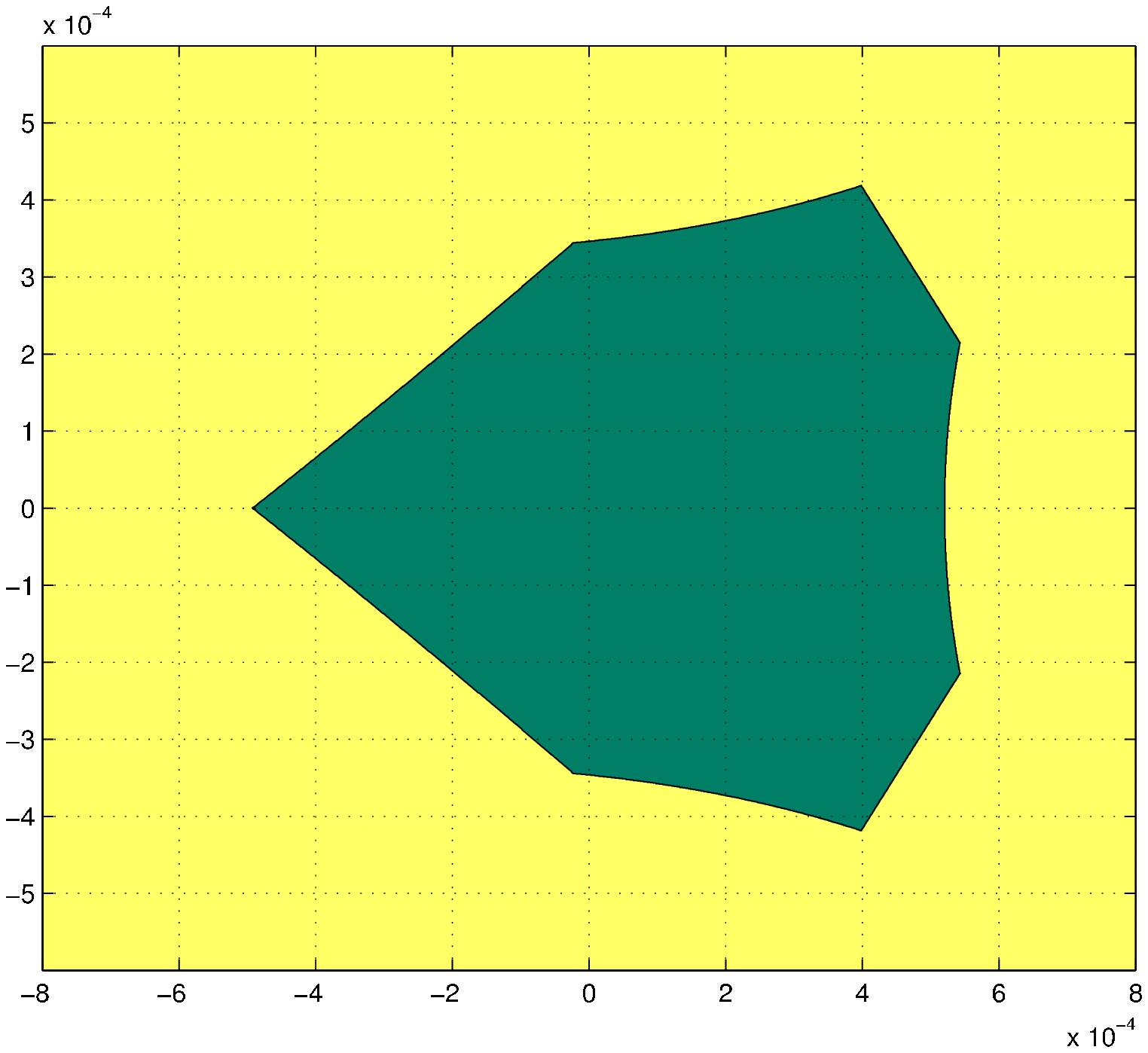}
}
\caption{Hyperbolicity region of Grad's $(M+1)$-moment system. The
$x$-axis is $g_{M-1}$ and the $y$-axis is $g_M$.}
\label{fig:hyp_reg}
\end{figure}

As a reference, the following corollary gives the characteristic
polynomial of the original matrix ${\bf A}_M$:
\begin{corollary} \label{cor:cp}
The characteristic polynomial of ${\bf A}_M$ is
\begin{equation} \label{eq:Grad_cp}
\theta^{\frac{M+1}{2}}
  \He_{M+1} \left( \frac{\lambda - u}{\sqrt{\theta}} \right)
  - \frac{(M+1)!}{2\rho} \left[
    \left( (\lambda - u)^2 - \theta \right) f_{M-1}
    + 2(\lambda - u) f_M
  \right].
\end{equation}
\end{corollary}
\begin{proof}
This can be shown by direct calculation:
\begin{equation}
\begin{split}
& |\lambda {\bf I} - {\bf A}_M| = \left|
  \lambda {\bf I} - (u{\bf I} + \sqrt{\theta} {\bf \Lambda}^{-1}
   \tilde{\bf A}_M {\bf \Lambda})
\right| \\
={} & \theta^{\frac{M+1}{2}} \left|
  \frac{\lambda - u}{\sqrt{\theta}} {\bf I} - {\bf \Lambda}^{-1}
   \tilde{\bf A}_M {\bf \Lambda}
\right| = \theta^{\frac{M+1}{2}} \left|
  \frac{\lambda - u}{\sqrt{\theta}} {\bf I} - \tilde{\bf A}_M
\right| \\
={} & \theta^{\frac{M+1}{2}} \Bigg\{
  \He_{M+1}\left( \frac{\lambda - u}{\sqrt{\theta}} \right) \\
& \qquad - \frac{(M+1)!}{2} \left[
    \left( \frac{(\lambda - u)^2}{\theta} - 1 \right)
    \frac{f_{M-1}}{\rho \theta^{(M-1)/2}} +
    \frac{2(\lambda - u)}{\sqrt{\theta}} \frac{f_{M}}{\rho \theta^{M/2}}
  \right]
\Bigg\} \\
={} & \theta^{\frac{M+1}{2}}
  \He_{M+1} \left( \frac{\lambda - u}{\sqrt{\theta}} \right)
  - \frac{(M+1)!}{2\rho} \left[
    \left( (\lambda - u)^2 - \theta \right) f_{M-1}
    + 2(\lambda - u) f_M
  \right].
\end{split}
\end{equation}
\end{proof}

\section{Hyperbolic moment system} \label{sec:hyperbolic} The loss of
global hyperbolicity of Grad's moment system has long been considered
as a failure of moment method. Recently, some encouraging progresses
are made in this direction \cite{Levermore, Torrilhon2010}.  However,
in the case that the number of moments is greater than $10$,
Levermore's method leads to great difficulties for the numerical
implementation, since the moments cannot be analytically solved from
the Lagrange multipliers\footnote{A local system is required to be
  solved by Newton iteration on each grid for every time step. We
  refer the readers to \cite{Tallec} for details. There is no report
  indicating that such a system has a fast solver.}; and it has been
demonstrated by Junk \cite{Junk} that the domain of definition for a
realizable distribution is not convex. Torrilhon's method mainly
focuses on 13-moment case in one space dimension, which seems not
trivial to be extended to the general cases. To the best of our
knowledge, no results for general moment system have been published.

In this section, we provide the method to regularize moment system
based on the results in Section \ref{sec:Grad} to achieve global
hyberbolicity. We discuss only 1D case here and the multi-dimensional
problems will be reported soon in later papers.

\subsection{Construction of hyperbolic moment system}
For an $(M+1)$-moment system containing quantities $\{\rho, u, \theta,
f_3, \cdots, f_M\}$, the Grad's moment system gives \emph{accurate}
evolution equations for most variables expect for $f_M$, since
$f_{M+1}$ appears in the accurate equation of $f_M$, and is forced to
be zero in Grad's closure. Almost all the regularization methods in
references are focused on the reconstruction of $f_{M+1}$, trying to
express $f_{M+1}$ as a function of the $M+1$ known variables in some
possible ways such as Chapman-Enskog expansion or realizing a positive
distribution \cite{Levermore, Struchtrup2003, NRxx_new,
  Torrilhon2010}. In this paper, we also limit our regularization to
the modification of the equation of $f_M$. However, since $f_{M+1}$
exists in this equation only in the form of its derivative, here we
directly substitute $\partial f_{M+1} / \partial x$ with some other
expression to gain global hyperbolicity.

Corollary \ref{cor:cp} shows that the characteristic polynomial of
${\bf A}_M$ is independent of $f_3$, $\cdots$, $f_{M-2}$, and its
dependence of $f_{M-1}$ and $f_M$ can be regarded as the result of
truncation. That is, if a Grad's system with $M+3$ or more variables
is considered, then $f_{M-1}$ and $f_M$ do not affect the
characteristic polynomial, either. Thus, it is reasonable to modify
the matrix ${\bf A}_M$ such that its characteristic polynomial is a
function only of $u$ and $\theta$. More precisely, the characteristic
polynomial of the modified matrix should always be
\begin{equation} \label{eq:modified_cp}
\theta^{\frac{M+1}{2}} \He_{M+1} \left(
  \frac{\lambda - u}{\sqrt{\theta}}
\right),
\end{equation}
which is obtained by substituting $f_{M-1} = f_M = 0$ into
\eqref{eq:Grad_cp}. Recalling that only the equation of $f_M$ is
allowed to be changed, we summarize all the requirements and raise the
following problem:
\begin{quote} \it%
Find $M+1$ functions $a_j = a_j(\bw_M)$, $j = 1,\cdots,M+1$, such that
\begin{displaymath}
\left|
  \lambda {\bf I} - {\bf A}_M - \sum_{j=1}^{M+1} a_j {\bf E}_{M+1,j}
\right| = \theta^{\frac{M+1}{2}} \He_{M+1} \left(
  \frac{\lambda - u}{\sqrt{\theta}}
\right), \quad \forall \rho,u,\theta,f_3,\cdots,f_M,
\end{displaymath}
where ${\bf E}_{ij}$ denotes the matrix $\be_i \be_j^T$, and $\be_j$
is the unit vector whose $j$-th component is equal to $1$.
\end{quote}
If $a_j = a_j(\bw_M)$, $j = 1,\cdots,M+1$ is the solution of this
problem, then a globally hyperbolic system can be obtained by
substituting the matrix ${\bf A}_M$ in \eqref{eq:Grad} with
\begin{equation} \label{eq:modified_matrix}
\hat{\bf A}_M := {\bf A}_M + \sum_{j=1}^{M+1} a_j {\bf E}_{M+1,j}.
\end{equation}
The rest part of this section will be devoted to tackling this
problem.

In order to simplify the notation, we use $S^{i,j}$ to denote the
$(i,j)$-th minor of the matrix $\lambda {\bf I} - {\bf A}_M$, and
define $S(k)$ as its \emph{$k$-th order leading principal minor},
which is the determinant of the upper-left part of $\lambda
{\bf I} - {\bf A}_M$ with $k$ rows and $k$ columns. According to the
expression of ${\bf A}_M$ \eqref{eq:A_M}, it is not difficult to find
\begin{subequations} \label{eq:S}
\begin{gather}
\label{eq:S12}
S^{M+1,1} = (-1)^M M!, \quad
  S^{M+1,2} = (-1)^{M-1} \frac{M!}{\rho} (\lambda - u), \\
\label{eq:S3}
S^{M+1,3} = (-1)^{M-2} \frac{M!}{\rho} [(\lambda - u)^2 - \theta], \\
S^{M+1,j} = (-1)^{M+1-j} \frac{M!}{(j-1)!} S(j-1),
  \quad j=4, \cdots, M+1.
\end{gather}
\end{subequations}
Now we expand the characteristic polynomial of the matrix \eqref
{eq:modified_matrix} as
\begin{equation}
\left|
  \lambda {\bf I} - {\bf A}_M - \sum_{j=1}^{M+1} a_j {\bf E}_{M+1,j}
\right| = |\lambda {\bf I} - {\bf A}_M| -
  \sum_{j=1}^{M+1} (-1)^{M+1+j} a_j S^{M+1,j}.
\end{equation}
In order that the above expression equals to \eqref{eq:modified_cp},
according to Corollary \ref{cor:cp}, we may choose $a_j$ such that
\begin{equation} \label{eq:equiv}
\frac{(M+1)!}{2\rho} \left[
  \left( (\lambda - u)^2 - \theta \right) f_{M-1}
  + 2(\lambda - u) f_M
\right] + \sum_{j=1}^{M+1} (-1)^{M+1-j} a_j S^{M+1,j} \equiv 0.
\end{equation}
The leading principal minor $S(k)$ is a polynomial in $\lambda$ of
degree $k$, since it is the characteristic polynomial of the $k\times
k$ upper-left block of ${\bf A}_M$. Hence, $S^{M+1,j}$ is a polynomial
in $\lambda$ of degree $j-1$, which can be observed from \eqref{eq:S}.
Such observation directly leads to
\begin{equation} \label{eq:a_j}
a_j \equiv 0, \quad j = 4,\cdots,M+1,
\end{equation}
since the first term in \eqref{eq:equiv} is a quadratic polynomial in
$\lambda$. Then, we put \eqref{eq:S12} and \eqref{eq:S3} into
\eqref{eq:equiv} and some simplification gives
\begin{equation}
[(\lambda - u)^2 - \theta] \left( \frac{M+1}{2} f_{M-1} + a_3 \right)
  + (\lambda - u) [(M+1) f_M + a_2] + a_1 \equiv 0.
\end{equation}
Now, the choices of $a_1$, $a_2$ and $a_3$ are naturally given as
\begin{equation} \label{eq:a1-3}
a_1 = 0, \quad a_2 = -(M+1) f_M, \quad a_3 = -\frac{M+1}{2} f_{M-1}.
\end{equation}
For simplicity, the notation $\RM$ is introduced as follows:
\begin{definition}
  The regularization term based on the characteristic speed correction
  is denoted as
  \begin{equation} \label{R_M} 
    \RM \triangleq \frac{M+1}{2} \left(2
      f_M \frac{\partial u}{\partial x} + f_{M-1} \frac{\partial
        \theta}{\partial x} \right).
  \end{equation}
\end{definition}
\noindent Then we have the following theorem:
\begin{theorem} \label{thm:hyperbolic}
The moment system
\begin{equation} \label{eq:modified_moment_system} 
  \frac{\partial \bw_M}{\partial t} + {\bf A}_M \frac{\partial 
    \bw_M}{\partial x} - \RM \be_{M+1} = 0
\end{equation}
is strictly hyperbolic if $\theta > 0$, and its characteristic speeds
are
\begin{equation} \label{eq:s_j}
s_j = u + c_j \sqrt{\theta}, \quad j = 1,\cdots, M+1,
\end{equation}
where $c_j$ is the $j$-th root of $\He_{M+1}(x)$.
\end{theorem}
\begin{proof}
The equations \eqref{eq:modified_moment_system} can be rewritten as
\begin{equation}
\frac{\partial \bw_M}{\partial t} +
  \left( {\bf A}_M + \sum_{j=1}^{M+1} a_j {\bf E}_{M+1,j} \right)
  \frac{\partial \bw_M}{\partial x} = 0,
\end{equation}
where $a_j$, $j=1,\cdots,M+1$ are defined in \eqref{eq:a_j} and
\eqref{eq:a1-3}. As we have discussed above, \eqref{eq:modified_cp}
gives the characteristic polynomial of the matrix in the parentheses,
which will be denoted by $\hat{\bf A}_M$ below as in
\eqref{eq:modified_matrix}. If $\theta > 0$, one has
\begin{equation}
|\lambda {\bf I} - \hat{\bf A}_M|
= \theta^{\frac{M+1}{2}} \He_{M+1}
  \left( \frac{s_j - u}{\sqrt{\theta}} \right)
= \theta^{\frac{M+1}{2}} \He_{M+1}(c_j) = 0.
\end{equation}
Therefore, \eqref{eq:s_j} gives all eigenvalues of $\hat{\bf A}_M$.
Since the Hermite polynomial $\He_{M+1}(x)$ has $M+1$ different zeros
in $\bbR$ \cite{Shen}, all $c_j$'s are distinct. Thus, the matrix
$\hat{\bf A}_M$ has no duplicate eigenvalues, hence is
diagonalizable. This indicates that \eqref{eq:modified_moment_system}
is a strictly hyperbolic system.
\end{proof}
Comparing with the exact moment system \eqref{eq:discrete_f}, the
hyperbolic system \eqref{eq:modified_moment_system} replaces $\partial
f_{M+1} / \partial x$ by
\begin{equation}
  -\frac{1}{M+1}\RM = -f_M \frac{\partial u}{\partial x} -
  \frac{1}{2} f_{M-1} \frac{\partial \theta}{\partial x}.
\end{equation}
This is a totally new way to regularize Grad's moment system.

\begin{remark}
  By modifying the last row of the matrix ${\bf A}_M$, the
  characteristic speeds can be appointed. Our regularization
  \eqref{eq:modified_moment_system} selects a special set of
  characteristic speeds \eqref{eq:s_j} such that they coincide with
  the Gauss-Hermite interpolation points. As discussed in
  \cite{Torrilhon2006}, the characteristic speeds can be viewed as a
  sort of discretization of the distribution function. Therefore, the
  system \eqref{eq:modified_moment_system} is similar as the ``shifted
  and scaled discrete velocity model'', with the expectation of
  spectral convergence when $M$ goes to infinity. Meanwhile, unlike
  the ordinary discrete velocity model, the nonlinearity of Grad's
  moment systems introduced by shifting and scaling of the basis
  functions is preserved. Additionally, such regularization is only a
  slight modification based on the original Grad's moment system, and
  we will find in the next subsection that a number of interesting
  properties can be obtained.
\end{remark}

\subsection{Characteristic waves of hyperbolic moment system}
In this part, we will focus on the Riemann problem of
\eqref{eq:modified_moment_system}. First, we claim that all
characteristic fields of \eqref{eq:modified_moment_system} is either
genuinely nonlinear or linearly degenerate. To verify this, we write
the right eigenvectors of $\hat{\bf A}_M$ in the following theorem:
\begin{theorem} \label{thm:eigenvector}
The right eigenvector of $\hat{\bf A}_M$ with eigenvalue $u + c_j
\sqrt{\theta}$ is
\begin{equation}
\br_j = (r_{j,1}, \cdots, r_{j,M+1})^T, \quad j=1,\cdots,M+1,
\end{equation}
where $c_j$ is the $j$-th root of Hermite polynomial $\He_{M+1}(x)$,
and $r_{j,k}$ is defined as
\begin{equation}\label{eq:eigenvector}
\begin{gathered}
r_{j,1} = \rho, \quad r_{j,2} = c_j \sqrt{\theta},
  \quad r_{j,3} = (c_j^2 - 1) \theta, \\
r_{j,k} = \frac{\He_{k-1}(c_j)}{(k-1)!} \rho \theta^{\frac{k-1}{2}}
  - \frac{c_j^2 - 1}{2} \theta f_{k-3} - c_j \sqrt{\theta} f_{k-2},
  \quad k = 4, \cdots, M+1.
\end{gathered}
\end{equation}
\end{theorem}
\begin{proof}
To prove this theorem, we need only to prove 
\begin{equation}\label{eq:eigen_eq}
\hat{\bf A}_M\br_j = (u+c_j\sqrt{\theta})\br_j.
\end{equation}
Split $\hat{\bf A}_M$ by row as 
$\hat{\bf A}_M = (\boldsymbol{a}_1^T,
            \boldsymbol{a}_2^T,\cdots,\boldsymbol{a}_{M+1}^T)^T$,
where $\boldsymbol{a}_k$ is the $k$-th row of $\hat{\bf A}_M$, $k =
1,2,\cdots,M+1$. Thus (\ref{eq:eigen_eq}) can be written as
\begin{equation}\label{eq:eigen_split}
(\boldsymbol{a}_1\br_j, \boldsymbol{a}_2\br_j,\cdots,
            \boldsymbol{a}_{M+1}\br_j)^T =
(u+c_j\sqrt{\theta})\br_j,\quad \text{for } j = 1,2,\cdots,M+1.
\end{equation}
With the expression of $\hat{\bf A}_M$, the first four rows of
\eqref{eq:eigen_split} can be verified directly:
\begin{equation} \label{eq:a1-4_rj}
\begin{split}
\boldsymbol{a}_1\br_j &= u r_{j,1} + \rho r_{j,2} =
    \rho(u+c_j\sqrt{\theta}) = r_{j,1}(u+c_j\sqrt{\theta}),\\
\boldsymbol{a}_2\br_j &= \theta/\rho \cdot r_{j,1}+u r_{j,2} + r_{j,3} =
    c_j\sqrt{\theta}(u+c_j\sqrt{\theta}) = r_{j,2}(u+c_j\sqrt{\theta}),\\
\boldsymbol{a}_3\br_j &= 2\theta r_{j,2}+u r_{j,3} + 6/\rho \cdot r_{j,4} =
    (c_j^2-1)\theta(u+c_j\sqrt{\theta}) = r_{j,3}(u+c_j\sqrt{\theta}),\\
\boldsymbol{a}_4\br_j &= 4f_3 r_{j,2} + \rho\theta/2 \cdot r_{j,3} + 
    u r_{j,4} + 4 r_{j,5} \\
        &= \He_3(c_j)\rho\theta^{3/2}/6\cdot (u+c_j\sqrt{\theta})
= r_{j,4}(u+c_j\sqrt{\theta}), \qquad \text{(only when $M \geqslant 4$)}.
\end{split}
\end{equation}
For $5 \le k \le M$,
\begin{equation} \label{eq:proof-eigenvector-j}
\begin{split}
\boldsymbol{a}_k\br_j = -\frac{\theta f_{k-2}}{\rho} r_{j,1} &+ k f_{k-1}
r_{j,2} + \frac{1}{2}[(k-2) f_{k-2} + \theta f_{k-4}] r_{j,3} \\
&-\frac{3 f_{k-3}}{\rho} r_{j,4} 
    + \theta r_{j,j-1} + u r_{j,j} + k r_{j,j+1}.
\end{split}
\end{equation}
Then, we substitute \eqref{eq:eigenvector} into
\eqref{eq:proof-eigenvector-j}, and get
\begin{equation} \label{eq:ak_rj}
\begin{split}
\boldsymbol{a}_k\br_j & = \frac{\He_{k-2}(c_j)}{(k-2)!}\theta^{k/2} 
  + u \frac{\He_{k-1}(c_j)}{(k-1)!}\theta^{(k-1)/2}
  + \frac{\He_{k}(c_j)}{(k-1)!}\theta^{k/2} \\
  & \qquad + (-c_j\theta^{1/2})(c_j\sqrt{\theta} + u )f_{k-2} 
  + [- \theta (c_j^2-1) (u + c_j \sqrt{\theta}) / 2] f_{k-3} \\
  & = u \frac{\He_{k-1}(c_j)}{(k-1)!}\theta^{(k-1)/2}
    + c_j\frac{\He_{k-1}(c_j)}{(k-1)!}\theta^{k/2}\\
  & \qquad + (-c_j\theta^{1/2})(c_j\sqrt{\theta} + u )f_{k-2} 
    + [- \theta (c_j^2-1) (u + c_j \sqrt{\theta}) / 2] f_{k-3} \\
  & = (u+c_j \sqrt{\theta}) r_{j,k}.
\end{split}
\end{equation}
For $k=M+1$, the situation is similar as $5 \le k \le M$. We expand
$\boldsymbol{a}_{M+1}\br_j$ as
\begin{equation}\label{eq:proof-eigenvector-M+1}
\begin{split}
\boldsymbol{a}_{M+1}\br_j = {}&
  \frac{\He_{M-1}(c_j)}{(M-1)!}\theta^{(M+1)/2} 
    + u \frac{\He_{M}(c_j)}{M!}\theta^{M/2}\\
    &+ (-c_j\theta^{1/2})(c_j\sqrt{\theta} + u )f_{M-1} 
    + [- \theta (c_j^2-1) (u + c_j \sqrt{\theta}) / 2] f_{M-2}.
\end{split}
\end{equation}
Here, $c_j$ is the $j$-th root of Hermite polynomial $\He_{M+1}(x)$.
Hence, the recursion relation of Hermite polynomials gives
\begin{equation}
\He_{M-1}(c_j) = \frac{c_j \He_M(c_j)}{M}.
\end{equation}
Substituting this equation into \eqref{eq:proof-eigenvector-M+1}, we
get 
\begin{equation} \label{eq:aM+1_rj}
\boldsymbol{a}_{M+1} \br_j = (u+c_j\sqrt{\theta})r_{j,M+1}.
\end{equation}
Collecting \eqref{eq:a1-4_rj}, \eqref{eq:ak_rj} and
\eqref{eq:aM+1_rj}, we finally arrive at \eqref{eq:eigen_split}. This
completes the proof of the theorem.
\end{proof}

\begin{corollary} \label{cor:gn_or_ld}
Each characteristic field of the hyperbolic system
\eqref{eq:modified_moment_system} is either genuinely nonlinear or
linearly degenerate.
\end{corollary}
\begin{proof}
Let $s_j = u + c_j \sqrt{\theta}$, and we only need to verify that
either $\nabla_{\bw_M} s_j \cdot \br_j \equiv 0$ or $\nabla_{\bw_M}
s_j \cdot \br_j \not\equiv 0$ holds. Since
\begin{equation}
\nabla_{\bw_M} s_j = \left(
  0, 1, \frac{1}{2} \frac{c_j}{\sqrt{\theta}}, 0, \cdots, 0
\right)^T,
\end{equation}
we have
\begin{equation}
\nabla_{\bw_M} s_j \cdot \br_j =
  c_j \sqrt{\theta} + \frac{1}{2} c_j (c_j^2 - 1) \sqrt{\theta}
= \frac{1}{2} c_j (c_j^2 + 1) \sqrt{\theta}.
\end{equation}
If $c_j$ is zero, the right hand side vanishes, while if $c_j$ is
nonzero, it is clear that $\nabla_{\bw_M} s_j \cdot \br_j \not\equiv
0$.
\end{proof}

This corollary indicates the simplicity of characteristic waves in the
solution of Riemann problems. Consider the following Riemann problem:
\begin{equation}
\begin{split}
& \frac{\partial \bw_M}{\partial t}
  + \hat{\bf A}_M \frac{\partial \bw_M}{\partial x} = 0, \\
& \bw_M(0,x) = \left\{ \begin{array}{ll}
  \bw_M^L, & x < 0, \\ [2mm]
  \bw_M^R, & x > 0. \\
\end{array} \right.
\end{split}
\end{equation}
A typical solution of this problem is the composition of at most $M+2$
intermediate states
\begin{displaymath}
\bw_M^0 = \bw_M^L, \quad \bw_M^1, \quad \cdots, \quad \bw_M^N, \quad
\bw_M^{N+1} = \bw_M^R, \qquad N \leqslant M,
\end{displaymath}
which are connected by $N+1$ elementary waves: rarefaction waves,
contact discontinuities, or shock waves. In order to get a full
understanding of the hyperbolic moment system, these waves will be
studied respectively below.

\subsubsection{Rarefaction waves}
As all hyperbolic systems, the integral curves and the Riemann
invariants are the major objects for the investigation of rarefaction
waves. The parameterization of an integral curve of the vector field
$\br_j$ satisfies
\begin{equation} \label{eq:int_curve}
\tilde{\bw}_M'(\zeta) = \br_j(\tilde{\bw}_M(\zeta)),
\end{equation}
where $\zeta$ is the parameter, and
\[
\tilde{\bw}_M(\zeta) =
\left(\tilde{\rho}(\zeta), \tilde{u}(\zeta), \tilde{\theta}(\zeta),
  \tilde{f}_3(\zeta), \cdots, \tilde{f}_M(\zeta)\right)^T
\]
denotes the integral curve in the $(M+1)$-dimensional phase space. For
a given point $\bw_M^0$ in the phase space, the integral curve through
$\bw_M^0$ can actually be analytically solved. Here we do not intend
to write down the complete expressions, while the analytical solutions
of $\rho(\zeta)$, $u(\zeta)$ and $\theta(\zeta)$ are given as
\begin{subequations} \label{eq:int_curve_pv}
\begin{align}
\label{eq:rarefactionwave_rho}
\tilde{\rho}(\zeta) & = \rho^0 \exp(\zeta), \\
\label{eq:rarefactionwave_u}
\tilde{u}(\zeta) & = u^0 + \frac{2c_j}{c_j^2-1} \sqrt{\theta^0} \left[
  \exp\left( \frac{c_j^2 - 1}{2} \zeta \right) - 1
\right],\\ 
\label{eq:rarefactionwave_theta}
\tilde{\theta}(\zeta) & = \theta^0 \exp\left( (c_j^2-1) \zeta \right).
\end{align}
\end{subequations}
It is easy to verify that \eqref{eq:int_curve_pv} satisfies the first
three equations of \eqref{eq:int_curve}. Note that in
\eqref{eq:int_curve}, only $\rho$, $\theta$ and $f_{j-2}$, $f_{j-1}$
appear in the right hand side of $f_j$'s equation, $j = 3,\cdots,M$.
Therefore, if the complete solution of $\tilde{\bw}_M(\zeta)$ is
needed, one can solve $f_j(\zeta)$ by explicit integration. Now we use
\eqref{eq:int_curve_pv} to give the $j$-th eigenvalue of ${\bf
  A}_M(\tilde{\bw}_M(\zeta))$ as
\begin{equation}\label{eq:rarefactionwave_sj}
s_j(\tilde{\bw}_M(\zeta)) =
  \tilde{u}(\zeta) + c_j \sqrt{\tilde{\theta}(\zeta)} =
  u^0 + c_j \sqrt{\theta^0} +
    \frac{c_j^2 + 1}{c_j^2 - 1} c_j \sqrt{\theta^0} \left[
      \exp\left( \frac{c_j^2 - 1}{2} \zeta \right) - 1
    \right].
\end{equation}
It is not difficult to prove that $s_j(\tilde{\bw}_M(\zeta)) \gtrless
s_j(\tilde{\bw}_M^0)$ if and only if $c_j \zeta \gtrless 0$, which is
helpful to predicate which part of the integral curve satisfies the
entropy condition. And substitution of \eqref{eq:rarefactionwave_u}
into \eqref{eq:rarefactionwave_sj} gives
\begin{equation}
s_j(\tilde{\bw}_M(\zeta)) - s_j(\tilde{\bw}_M^0) = 
    \frac{c_j^2+1}{2}(\tilde{u}(\zeta) - u^0).
\end{equation}
Hence, $s_j(\tilde{\bw}_M(\zeta)) \gtrless s_j(\tilde{\bw}_M^0)$ holds
if and only if $\tilde{u}(\zeta) \gtrless u^0$. Therefore, if the left
state $\bw_M^L$ and the right state $\bw_M^R$ are connected by a
single rarefaction wave, $u^L < u^R$ has to be satisfied, since the
entropy condition requires $s_j(\bw_M^L) < s_j(\bw_M^R)$. Now let us
turn to the pressure $p$. Equations
\eqref{eq:rarefactionwave_rho} and \eqref{eq:rarefactionwave_theta}
show that
\begin{equation}
\tilde{p}(\zeta) = \tilde{\rho}(\zeta)\tilde{\theta}(\zeta) = p^0
    \exp(c_j^2\zeta).
\end{equation}
Therefore, the pressures on both sides of a rarefaction wave should
satisfy
\[\begin{array}{l} 
p^L < p^R,\quad \mathrm{if~} c_j > 0; \\
p^L > p^R,\quad \mathrm{if~} c_j <0.
\end{array}\]
Here we point out that the sign of $c_j$ is as
\begin{equation}
c_j \left\{ \begin{array}{ll}
  > 0, & \quad \text{if~} j > (M+1)/2, \\
  = 0, & \quad \text{if~} j = (M+1)/2, \\
  < 0, & \quad \text{if~} j < (M+1)/2.
\end{array} \right.
\end{equation}

It is interesting that Riemann invariants exist for all genuinely
nonlinear fields, and the following theorem gives its expressions.
\begin{theorem} 
For hyperbolic moment system \eqref{eq:modified_moment_system}, the
Riemann invariants for the $j$-family are
\begin{equation} \label{eq:Riemann_invariants}
\begin{split}
R_1 &= \rho \theta^{-1/(c_j^2-1)}, \quad
  R_2 = u - \frac{2c_j}{c_j^2-1} \sqrt{\theta}, \\
R_k &= C_{k,0} \rho \theta^{k/2} +
  \sum_{i=3}^k C_{k,i} f_i \theta^{(k-i)/2}, \quad k = 3, \cdots, M.
\end{split}
\end{equation}
where $C_{k,i}$ is defined recursively as
\begin{subequations}
\begin{align}
\label{eq:C_kk}
& C_{k,k} = 1, \quad C_{k,k-1} = \frac{2c_j}{c_j^2-1}, \\
\label{eq:C_ki}
& C_{k,i} = \frac{1}{k-i} \left(
  C_{k,i+2} + C_{k,i+1} \frac{2c_j}{c_j^2-1}
\right), \qquad i = 3,\cdots,k-2, \\
\label{eq:C_k0}
& C_{k,0} = \frac{2}{(1-c_j^2)k - 2} \sum_{i=3}^k
  \frac{\He_i(c_j)}{i!} C_{k,i}.
\end{align}
\end{subequations}
\end{theorem}
\begin{proof}
We only need to prove
\begin{equation} \label{eq:GradR_rj_eq_0}
\nabla_{\bw_M} R_k \cdot \br_j \equiv 0, \quad \forall k = 1,\cdots,M.
\end{equation}
The verification in the cases $k=1$ and $k=2$ is straightforward:
\begin{align}
\nabla_{\bw_M} R_1 \cdot \br_j & = \theta^{-1/(c_j^2-1)} \cdot \rho
  - \frac{1}{c_j^2 - 1} \rho \theta^{-1/(c_j^2-1)-1}
  \cdot (c_j^2-1) \theta = 0, \\
\nabla_{\bw_M} R_2 \cdot \br_j & = 1 \cdot c_j \sqrt{\theta} -
  \frac{c_j}{(c_j^2 - 1) \sqrt{\theta}} \cdot (c_j^2-1) \theta = 0.
\end{align}
If $k \geqslant 3$, the gradient of $R_k$ is
\begin{equation}
\begin{split}
\nabla_{\bw_M} R_k &= \Bigg(
  C_{k,0} \theta^{k/2}, \; 0, \;
  \frac{k}{2} C_{k,0} \rho \theta^{(k-1)/2} + 
    \sum_{i=3}^{k-1} \frac{k-i}{2} C_{k,i} f_i \theta^{(k-i)/2-1}, \\
& \qquad \qquad \qquad \qquad
  C_{k,3} \theta^{(k-3)/2}, \; C_{k,4} \theta^{(k-4)/2}, \;
  \cdots, \; C_{k,k} \theta^{(k-k)/2}
\Bigg)^T.
\end{split}
\end{equation}
With some rearrangement, $\nabla_{\bw_M} R_k \cdot \br_j$ is
simplified as
\begin{equation} \label{eq:GradR_rj}
\begin{split}
\nabla_{\bw_M} R_k \cdot \br_j & = \left[
  \left( 1 + \frac{k}{2} (c_j^2 - 1) \right) C_{k,0} +
  \sum_{i=3}^k C_{k,i} \frac{\He_i(c_j)}{i!}
\right] \rho \theta^{k/2} \\
& \quad + \sum_{i=3}^{k-2} \left[
  \frac{c_j^2 - 1}{2} (k-i) C_{k,i} -
  c_j C_{k,i+1} - \frac{c_j^2-1}{2} C_{k,i+2}
\right] f_i \theta^{(k-i)/2} \\
& \quad + \left(
  \frac{c_j^2-1}{2} C_{k,k-1} - c_j C_{k,k}
\right) f_{k-1}.
\end{split}
\end{equation}
We have that 
\begin{itemize}
\item \eqref{eq:C_kk} indicates that the last line of \eqref{eq:GradR_rj} is zero; 
\item \eqref{eq:C_ki} indicates that the second line of \eqref{eq:GradR_rj} is zero; 
\item \eqref{eq:C_k0} indicates that the first line of \eqref{eq:GradR_rj} is zero. 
\end{itemize}
Thus \eqref{eq:GradR_rj_eq_0} is proved.
\end{proof}

\subsubsection{Contact discontinuities}
According to the proof of Corollary \ref{cor:gn_or_ld}, the contact
discontinuities can only be found in the case of $c_j = 0$. Thus, if
$M$ is odd, no contact discontinuities exist in the characteristic
waves. For contact discontinuities, the discussion on integral curves
and Riemann invariants above is still valid. If we substitute $c_j =
0$ into \eqref{eq:Riemann_invariants}, $u$, $p$ and $f_3$ can be found
to be invariant acrossing the contact discontinuity.

\subsubsection{Shock waves}
Discussion of the shock waves requires additional scrupulosity. As been
well known, the jump condition on the shock wave is sensitive to the
form of the hyperbolic equations. Therefore, before we give the
Rankine-Hugoniot condition, it is necessary to rewrite
\eqref{eq:modified_moment_system} in an appropriate form. Though a
conservative form is desired, the whole system can no longer be
written as a conservation law since two terms are added to the last
equation. Nevertheless, the conservative form of the first $M$
equations remains. Thus \eqref{eq:modified_moment_system} can actually
be reformulated by $M$ conservation laws and a single non-conservative
equation. Precisely, if we let
\begin{equation}
\bq = (q_0, \cdots, q_M)^T, \qquad
  q_j = \frac{1}{j!} \int_{\bbR} \xi^j f(\xi) \dd \xi, \quad
  j = 0, \cdots, M,
\end{equation}
\eqref{eq:modified_moment_system} is reformulated as
\begin{equation} \label{eq:con_eq}
\begin{gathered}
\frac{\partial q_j}{\partial t} +
  (j+1) \frac{\partial q_{j+1}}{\partial x} = 0,
  \quad j = 0,\cdots,M-1, \\
\frac{\partial q_M}{\partial t} +
  \frac{\partial F(\bq)}{\partial x}
  - \RM = 0.
\end{gathered}
\end{equation}
The relation between $\bq$ and $\bw_M$ is
\begin{equation}
f_j = \sum_{k=0}^j (-1)^{j-k}
  \frac{\He_{j-k}(u/\sqrt{\theta})}{(j-k)!} \theta^{\frac{j-k}{2}} q_k,
\quad u = q_1 / q_0, \quad \theta = 2q_2 / q_0 - (q_1 / q_0)^2,
\end{equation}
and $F(\bq)$ is defined as
\begin{equation}
F(\bq) = (M+1) \sum_{k=0}^M
  (-1)^{M-k} \frac{\He_{M+1-k}(u/\sqrt{\theta})}{(M+1-k)!}
  \theta^{\frac{M+1-k}{2}} q_k.
\end{equation}
For convenience, we write \eqref{eq:con_eq} in the following form:
\begin{equation} \label{eq:noncon_system}
\frac{\partial \bq}{\partial t} +
  {\bf B}(\bq) \frac{\partial \bq}{\partial x} = 0,
\end{equation}
where ${\bf B}(\bq)$ is an $(M+1)\times (M+1)$ matrix.

Since \eqref{eq:noncon_system} is still a non-conservative system, the
DLM theory \cite{Maso} is introduced when discussing the shock wave. A
shock wave is a single jump discontinuity connecting two constant
states $\bq^L$ and $\bq^R$ in a genuinely nonlinear field $j$, and
$\bq^L$, $\bq^R$ and the propagation speed of the shock wave $S_j$
should satisfy the following conditions:
\begin{itemize}
\item
Generalized Rankine-Hugoniot condition:
\begin{equation}\label{eq:RHcondition}
\int_0^1 \left[
  S_j {\bf I} - {\bf B} \left( \bPhi(\nu; \bq^L, \bq^R) \right)
\right] \frac{\partial \bPhi}{\partial
            \nu}(\nu;\bq^L,\bq^R) \dd \nu = 0,
\end{equation}
where ${\bf I}$ is the identity matrix of order $M+1$, and $\bPhi(\nu;
\bq^L, \bq^R)$ is a locally Lipschitz mapping satisfying
\begin{equation}
\bPhi(0; \bq^L, \bq^R) = \bq^L, \quad \bPhi(1; \bq^L, \bq^R) = \bq^R.
\end{equation}
We refer the readers to \cite{Maso} for details. In Section
\ref{sec:col}, we will point out that the setup of $\bPhi$ is not
crucial if the collision term presents.
\item
Entropy condition
\begin{equation}\label{eq:shockwave_entropy}
s_j(\bq^L) > S_j > s_j(\bq^R).
\end{equation}
\end{itemize}
It is obvious that the first $M$ rows of \eqref{eq:RHcondition} are
independent of $\bPhi$; they are the same as the classical
Rankine-Hugoniot conditions. This allows us to analyze the properties
of the shock waves without regarding the form of $\bPhi$.

The first and second equations of \eqref{eq:RHcondition} can be
written as
\begin{align}
\rho^L u^L - \rho^R u^R &= S_j(\rho^L - \rho^R),
    \label{eq:shockwave_first}\\
\rho^L (u^L)^2 + \rho^L \theta^L - \rho^R (u^R)^2 - \rho^R\theta^R &=
S_j(\rho^L u^L - \rho^R u^R).\label{eq:shockwave_second}
\end{align}
Since $\rho^L \neq \rho^R$ and $\rho^L u^L \neq \rho^R u^R$ (otherwise
$\bq^L = \bq^R$), one has
\begin{subequations} \label{eq:shockwave_Si}
\begin{align}
S_j &= \frac{\rho^L u^L - \rho^R u^R}{\rho^L -
    \rho^R}\label{eq:shockwave_si1}\\
    & = \frac{\rho^L (u^L)^2 + \rho^L \theta^L - \rho^R (u^R)^2 -
        \rho^R\theta^R}{\rho^L u^L - \rho^R
            u^R}.\label{eq:shockwave_si2}
\end{align}
\end{subequations}
Putting \eqref{eq:shockwave_si1} into \eqref{eq:shockwave_entropy},
and multiplying both sides with $(\rho^L - \rho^R)^2$, we get
\begin{subequations} \label{eq:entropy_inequality}
\begin{align}
\rho^L(u^L-u^R)(\rho^L-\rho^R) &>
c_j(\rho^L-\rho^R)^2\sqrt{\theta^R},\label{eq:shockwave_entropy_R}\\
\rho^R(u^L-u^R)(\rho^L-\rho^R) &<
c_j(\rho^L-\rho^R)^2\sqrt{\theta^L}\label{eq:shockwave_entropy_L}.
\end{align}
\end{subequations}
If $c_j>0$, \eqref{eq:shockwave_entropy_R} gives
\begin{equation} \label{eq:shockwave_rhou}
(u^L-u^R)(\rho^L-\rho^R) > 0.
\end{equation}
Thus, we can divide both sides of \eqref{eq:entropy_inequality} by
$(u^L - u^R)(\rho^L - \rho^R)$ without changing the inequality sign,
and the result is:
\begin{equation}
\frac{\rho^L}{\sqrt{\theta^R}} >
  \frac{c_j (\rho^L - \rho^R)}{u^L - u^R} >
  \frac{\rho^R}{\sqrt{\theta^L}},
\end{equation}
from which one directly has
\begin{equation}\label{eq:shockwave_rhomultp1}
(\rho^L)^2 \theta^L - (\rho^R)^2 \theta^R > 0.
\end{equation}
Similarly, if $c_j<0$, we have
\begin{equation}\label{eq:shockwave_rhomultp2}
(u^L - u^R)(\rho^L - \rho^R) < 0,
\quad \text{and} \quad
(\rho^L)^2 \theta^L - (\rho^R)^2 \theta^R < 0.
\end{equation}

\begin{lemma}\label{eq:shockwave_u}
For hyperbolic moment system \eqref{eq:noncon_system}, if $\bq^L$
and $\bq^R$ are connected by a $j$-shock wave, then the following
inequalities hold:
\begin{align}
u^L &> u^R, \qquad \text{and } \qquad \left\{
    \begin{array}{ll}
        p^L > p^R, & \quad \text{if } c_j > 0,\\
        p^L < p^R, & \quad \text{if } c_j < 0.
    \end{array}\right.
\end{align}

\end{lemma}
\begin{proof}
With some rearrangement, \eqref{eq:shockwave_Si} can be reformulated
as
\begin{equation}\label{eq:shockwave_rhop}
(\rho^L - \rho^R)(\rho^L\theta^L - \rho^R\theta^R) =
\rho^L \rho^R (u^L-u^R)^2.
\end{equation}
Since the right hand side of \eqref{eq:shockwave_rhop} is positive,
one but only one of the following two statements is true:
\begin{enumerate}
\item $\rho^L > \rho^R$ and $\rho^L \theta^L > \rho^R \theta^R$;
\item $\rho^L < \rho^R$ and $\rho^L \theta^L < \rho^R \theta^R$.
\end{enumerate}
If $c_j > 0$, equation \eqref{eq:shockwave_rhomultp1} indicates that
the first statement is true. Then, we can use
\eqref{eq:shockwave_rhou} to conclude $u^L > u^R$. The conclusion for
the case $c_j < 0$ can be proved in the same way.
\end{proof}

Now, we summarize all our discussions on the entropy conditions of
three types of waves in the following theorem:
\begin{theorem}
  For hyperbolic moment system \eqref{eq:modified_moment_system},
  if the wave of the $j$-th family is elementary, then its type can be
  determined by the value of $c_j$ and the macroscopic velocities or
  pressures on both sides of the wave:
\begin{center}
\begin{tabular}{|l|c|c|}\hline
                                     &   \rm{Velocity}                        &   \rm{Pressure}      \\ \hline
\rm{Contact discontinuity}           &   $c_j = 0$, $u^L = u^R$               &   $c_j=0$, $p^L=p^R$ \\ \hline
\multirow{2}*{\rm{Rarefaction wave}} &  \multirow{2}*{$c_j\neq 0$, $u^L<u^R$} &   $c_j>0$, $p^L<p^R$ \\ \cline{3-3}
                                     &                                        &   $c_j<0$, $p^L>p^R$ \\ \hline
\multirow{2}*{\rm{Shock wave}}       &  \multirow{2}*{$c_j\neq 0$, $u^L>u^R$} &   $c_j>0$, $p^L>p^R$ \\ \cline{3-3}
                                     &                                        &   $c_j<0$, $p^L<p^R$ \\ \hline
\end{tabular}
\end{center}
\end{theorem}

\begin{remark}
  It is not difficult to find that Euler equations are a special case
  of the proposed hyperbolic moment equations. In the case of $M=2$,
  we have $f_1 = f_2 = 0$ thus the regularization vanishes. In other
  words, just like Grad's moment sytem, the hyperbolic system can be
  viewed as an extension of Euler equations. Actually, all the
  discussions in this section, including the eigenvalues and
  eigenvectors, Riemann invariants, and the entropy condition, are
  valid for the 1D Euler equations with adiabatic index $\gamma = 3$,
  while Grad's moment system is not able to preserve these criterions.
  In this respect, comparing with Grad's moment system, this
  regularized moment system is likely to be a more natural extension
  of Euler equations.
\end{remark}


\section{The case with collision terms} \label{sec:col}
In this section, we will give a short discussion on the moment
system with collision terms. For simplicity, the BGK collision
operator \cite{BGK} is considered. In this case, the Boltzmann
equation \eqref{eq:Boltzmann} becomes
\begin{equation}
\frac{\partial f}{\partial t} +
  \xi \frac{\partial f}{\partial x} = \frac{1}{\tau} (f_M - f),
\end{equation}
where $\tau$ is the relaxation time, and $f_M$ is the Maxwellian:
\begin{equation}
f_M = \frac{\rho}{\sqrt{2\pi \theta}} \exp \left(
  - \frac{|\xi - u|^2}{2 \theta}
\right).
\end{equation}
This equation leads to a very simple form of the collision term in the
moment system as
\begin{equation} \label{eq:collisional} 
  \frac{\partial \bw_M}{\partial t} + {\bf A}_M \frac{\partial \bw_M}
  {\partial x} - \delta_H \RM \be_{M+1} = -\frac{1}{\tau} {\bf P} \bw_M,
\end{equation}
where $\bf P$ is a diagonal matrix
\begin{equation} \label{eq:P}
{\bf P} = \mathrm{diag} \{ 0,0,0,1,\cdots,1 \},
\end{equation}
and $\delta_H = 0$ corresponds to Grad's moment system, while
$\delta_H = 1$ corresponds to the regularized moment system. Note that
when considering the weak solution of \eqref{eq:collisional}, one
still needs to rewrite \eqref{eq:collisional} as equations of $\bq$:
\begin{equation}
\begin{gathered}
\frac{\partial q_j}{\partial t} +
  (j+1) \frac{\partial q_{j+1}}{\partial x} =
  -\frac{1}{\tau} \mathcal{P}_j(q_0, \cdots, q_j),
  \quad j = 0,\cdots,M-1, \\
\frac{\partial q_M}{\partial t} +
  \frac{\partial F(\bq)}{\partial x}
  - \RM = -\frac{1}{\tau} \mathcal{P}_M(q_0, \cdots, q_M),
\end{gathered}
\end{equation}
where $\mathcal{P}_j$, $j=0,\cdots,M$ are the corresponding production
terms. Then the first order derivative part of the last equation will
still be treated using the DLM theory.

An important index that exhibits the quality of a collisional moment
system is its order of accuracy in term of $\tau$. The conception of
``order of accuracy'' is based on the assumption that $\tau$ is a
small quantity, and its precise definition can be found in
\cite{Struchtrup, Li}. In \cite{NRxx_new}, the order of magnitude for
each moment has been deduced as
\begin{equation} \label{eq:order}
f_k \sim O(\tau^{\lceil k/3 \rceil}), \quad k \geqslant 3
\end{equation}
for the infinite moment system, which is obtained by the technique of
Maxwellian iteration. It is easy to find that \eqref{eq:order} remains
correct for the regularized moment system (equation
\eqref{eq:collisional} with $\delta_H = 1$), since the order of $\RM$
never exceeds the leading order term of $f_M / \tau$. However, when $M
= 3m + 1$, $m \geqslant 1$, the order of accuracy of the moment system
is actually reduced by $2$ with presentation of the regularization
terms. This fact is not difficult to obtain and will be reported
elsewhere. In general, the order of accuracy still goes to infinity as
$M$ increases.

Another issue is the choice of the path function $\bPhi$, which was
introduced in \eqref{eq:RHcondition}. Let us restrict our discussion
of its role in solving a Riemann problem of
\eqref{eq:collisional}. First, we need to get some knowledge about the
general behavior of the solution, referring to the careful study of
the Riemann problem of 13-moment system in \cite{Torrilhon2000}.
Roughly speaking, the solution shows a number of waves initially, then
these waves are damping gradually, and eventually the solution tends
to a smooth curve which is similar as the solution of Euler equations.
The initial waves have no physical meanings due to the strong
non-equilibrium which cannot be described by the moment system, while
the solution gets close to the Boltzmann equation's solution only when
the waves are fully dissipated. Later, this behavior is verified
numerically for large number moment equations in \cite{Weiss}, where
the authors show that the speed of dissipation increases when the
number of moments gets larger. It is expected that this also describes
the evolution of regularized moment system. Based on
\cite{Torrilhon2000, Weiss}, we have the following assertions for the
regularized moment system:
\begin{enumerate}
\item If subshocks appear in the solution, the choice of $\bPhi$
  indeed makes sense. In this situation, the system is inadequate for
  the description of the physical process, saying $M$ needs to be
  increased.
\item $\bPhi$ affects the solution when the time $t$ is very
  small. However, such solution has no physical significance, either.
  Only when the solution gets close enough to a smooth function, the
  moment system starts to show its ability to describe physics. Note
  that the smooth solution is independent of $\bPhi$; therefore,
  $\bPhi$ only affects the way in which the waves are damped, but does
  not affect the intrinsic constituent of the solution.
\end{enumerate}
These two assertions indicate that the choice of $\bPhi$ is not
crucial in solving a Riemann problem. We can simply use a linear
function to connect any two states such that the numerical schemes can
be constructed easily.


\section{Numerical experimentation for a shock tube problem}
\label{sec:num}
In this section, a shock tube problem is studied numerically to show
the behavior of the hyperbolic moment systems. We consider the
following Riemann problem:
\begin{equation}
\begin{split}
& \frac{\partial \bw_M}{\partial t}
  + \hat{\bf A}_M \frac{\partial \bw_M}{\partial x}
  = -\frac{1}{\tau} {\bf P} \bw_M, \\
& \bw_M(0,x) = \left\{ \begin{array}{ll}
  \bw_M^L, & x < 0, \\ [2mm]
  \bw_M^R, & x > 0, \\
\end{array} \right.
\end{split}
\end{equation}
where $\bf P$ is defined in \eqref{eq:P} and the initial left and
right states are
\begin{equation}
\bw_M^L = (7,0,1,0,\cdots,0)^T, \quad
\bw_M^R = (1,0,1,0,\cdots,0)^T.
\end{equation}
The relaxation time is chosen as $\tau = \Kn / \rho$. Here two
different cases $\Kn = 0.05$ and $\Kn = 0.5$ are considered. A
nonconservative version of the HLL scheme \cite{Rhebergen} is employed
to discretize the moment system.

The numerical results for $\Kn = 0.05$ with $M$ ranging from $2$ to
$10$ are listed in Figure \ref{fig:Kn=0.05}, in which the thin black
lines are the numerical results of the hyperbolic moment equations
(HME), and the thick gray lines are the results of Mieussens' discrete
velocity model (DVM) \cite{Mieussens}, provided as reference
solutions. The profiles of $\rho$, $u$ and $p$ are drawn. It is clear
that the solutions of hyperbolic moment systems converge to the
solution of the Boltzmann equation when $M$ increases. Note that when
$M=2$, the hyperbolic moment system is equivalent to the Euler
equations, and the contact discontinuities and the shocks are obvious.
When $M = 3$, a shock can still be found near $x = 0.75$. When $M$ is
greater than $5$, the discontinuities are fully damped. This agrees
with Torillhon's theory \cite{Weiss} that the discontinuities are
damped faster when $M$ is larger.
\begin{figure}[!ht]
\psfrag{rho, CDVM}{\scalebox{.5}{$\rho$, DVM}}
\psfrag{rho, HME}{\scalebox{.5}{$\rho$, HME}}
\psfrag{u, CDVM}{\scalebox{.5}{$u$, DVM}}
\psfrag{u, HME}{\scalebox{.5}{$u$, HME}}
\psfrag{p, CDVM}{\scalebox{.5}{$p$, DVM}}
\psfrag{p, HME}{\scalebox{.5}{$p$, HME}}
\subfigure[$M=2$ (Euler)]{
\includegraphics[width=.31\textwidth,clip]{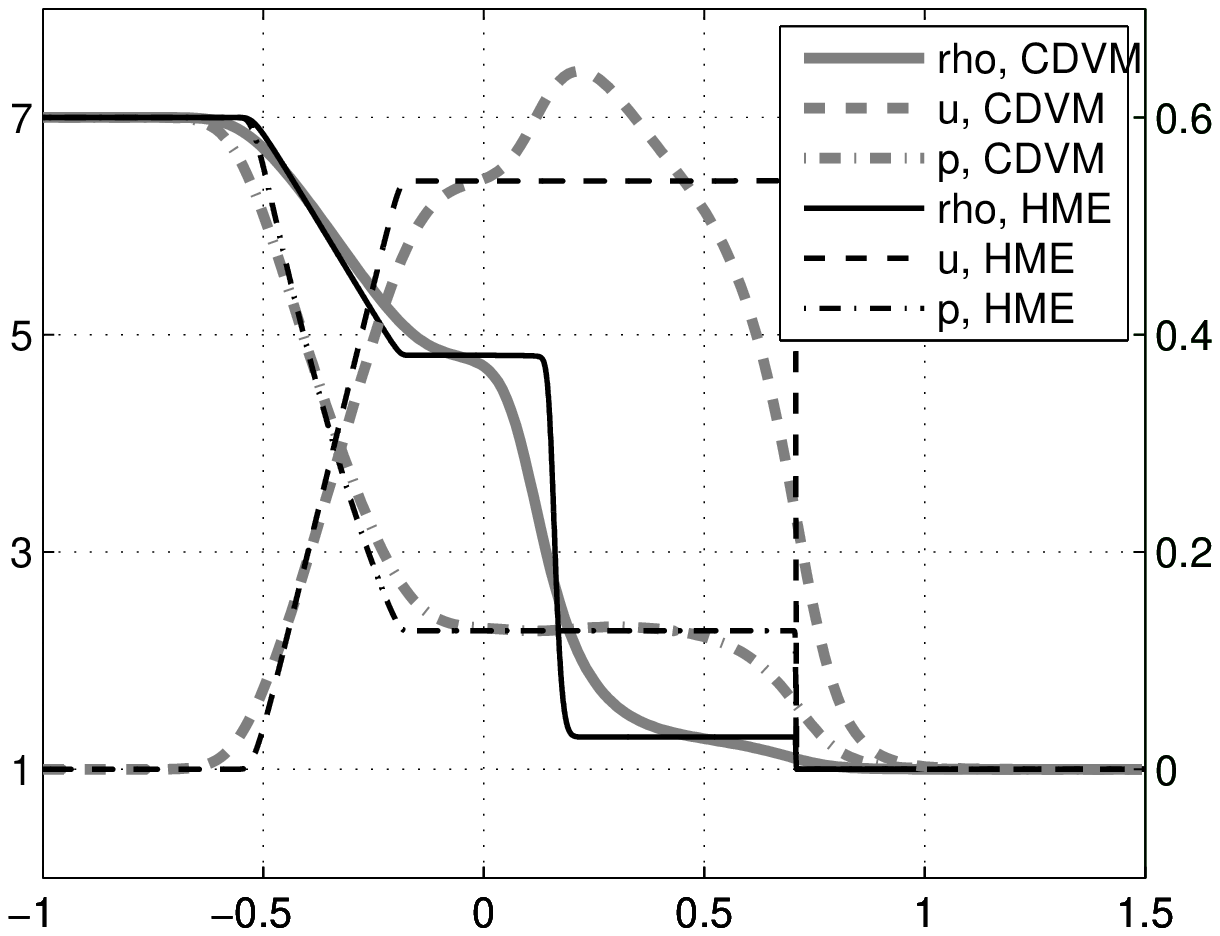}
}
\subfigure[$M=3$]{
\includegraphics[width=.31\textwidth,clip]{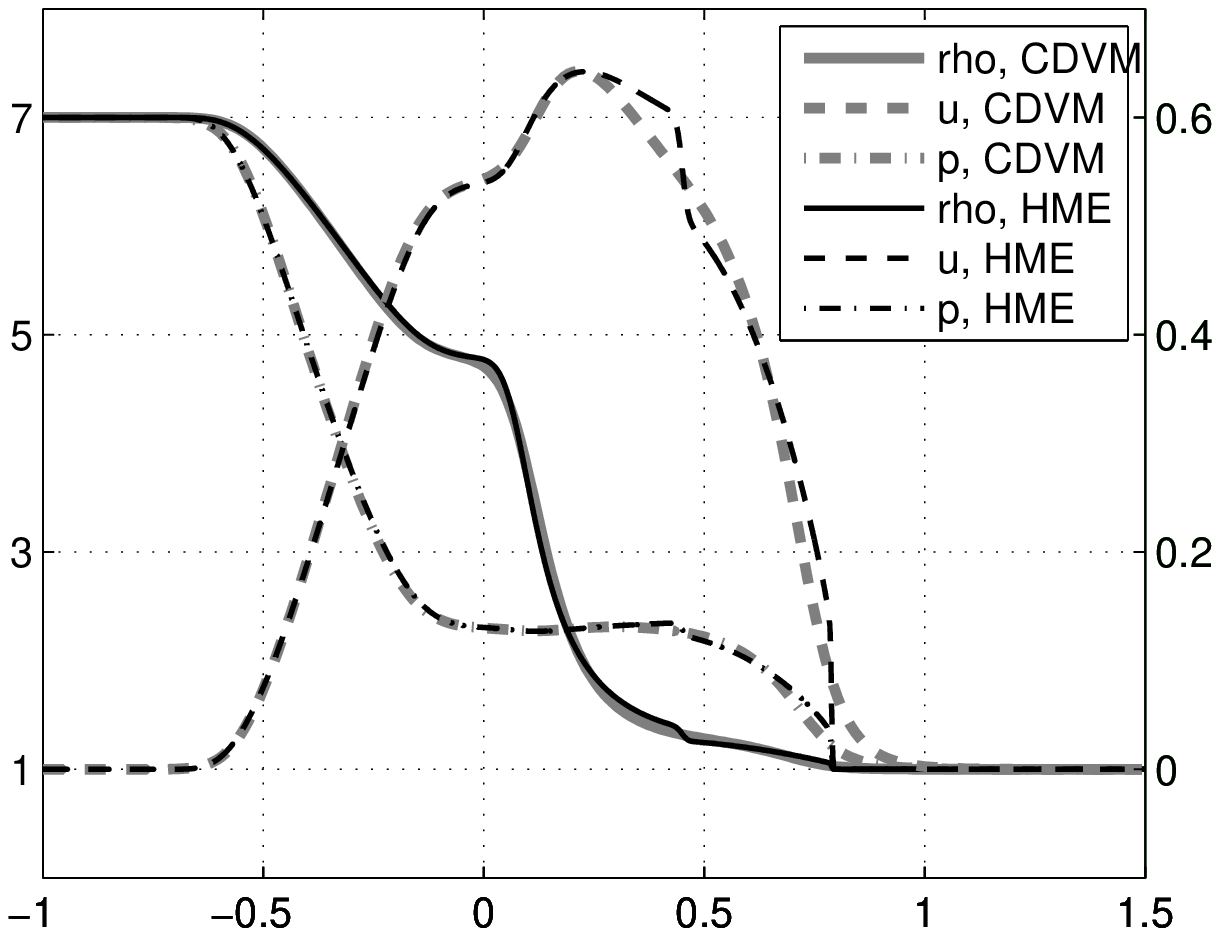}
}
\subfigure[$M=4$]{
\includegraphics[width=.31\textwidth,clip]{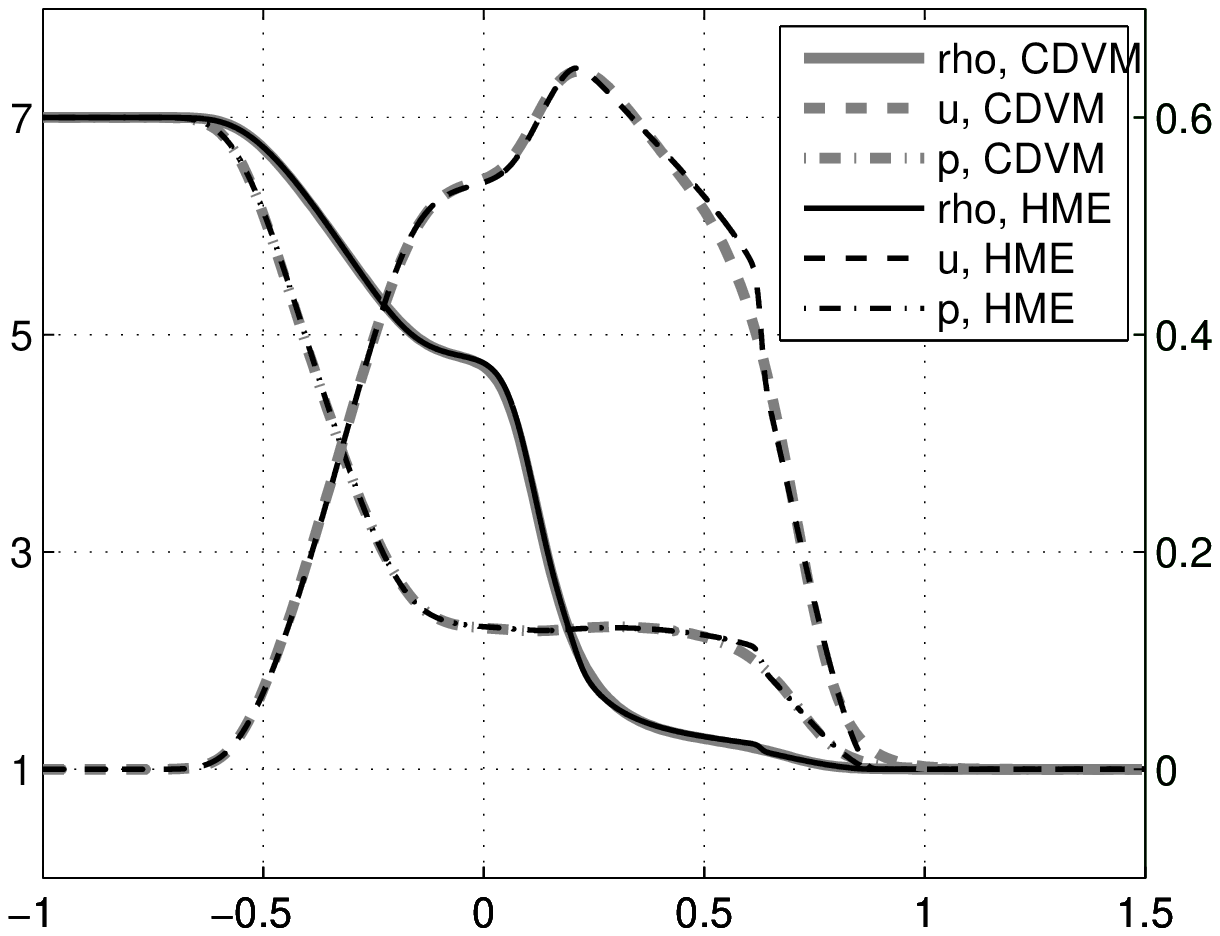}
}
\subfigure[$M=5$]{
\includegraphics[width=.31\textwidth,clip]{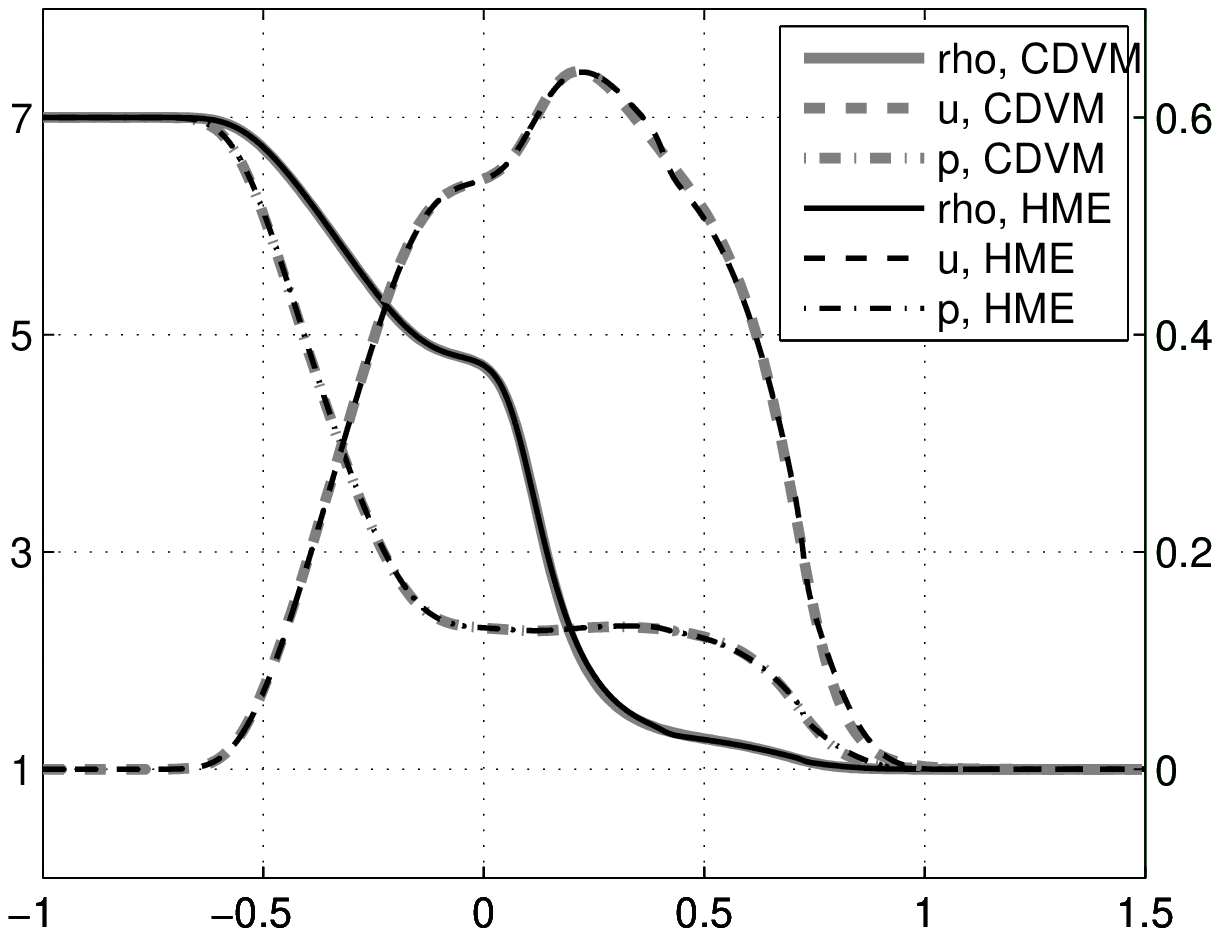}
}
\subfigure[$M=6$]{
\includegraphics[width=.31\textwidth,clip]{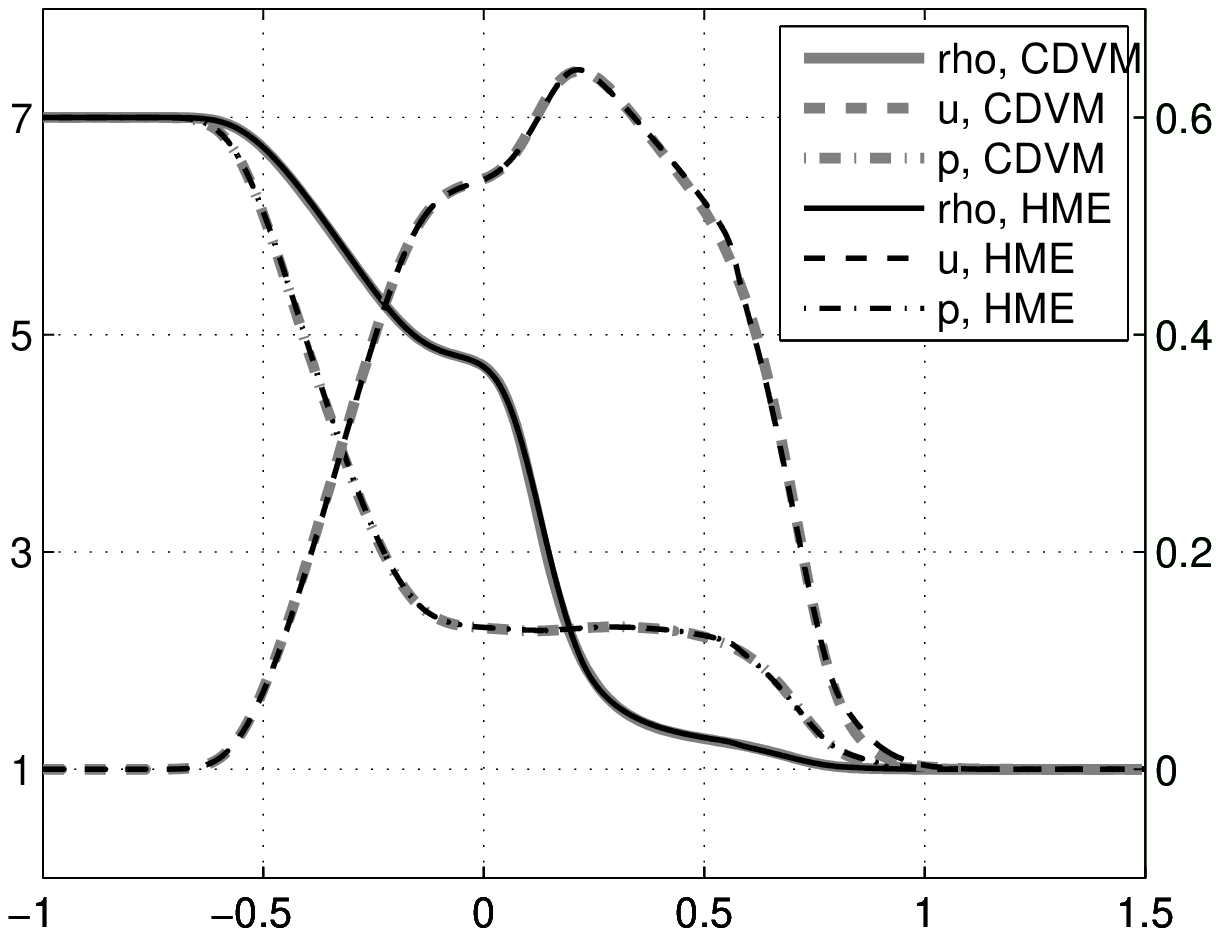}
}
\subfigure[$M=7$]{
\includegraphics[width=.31\textwidth,clip]{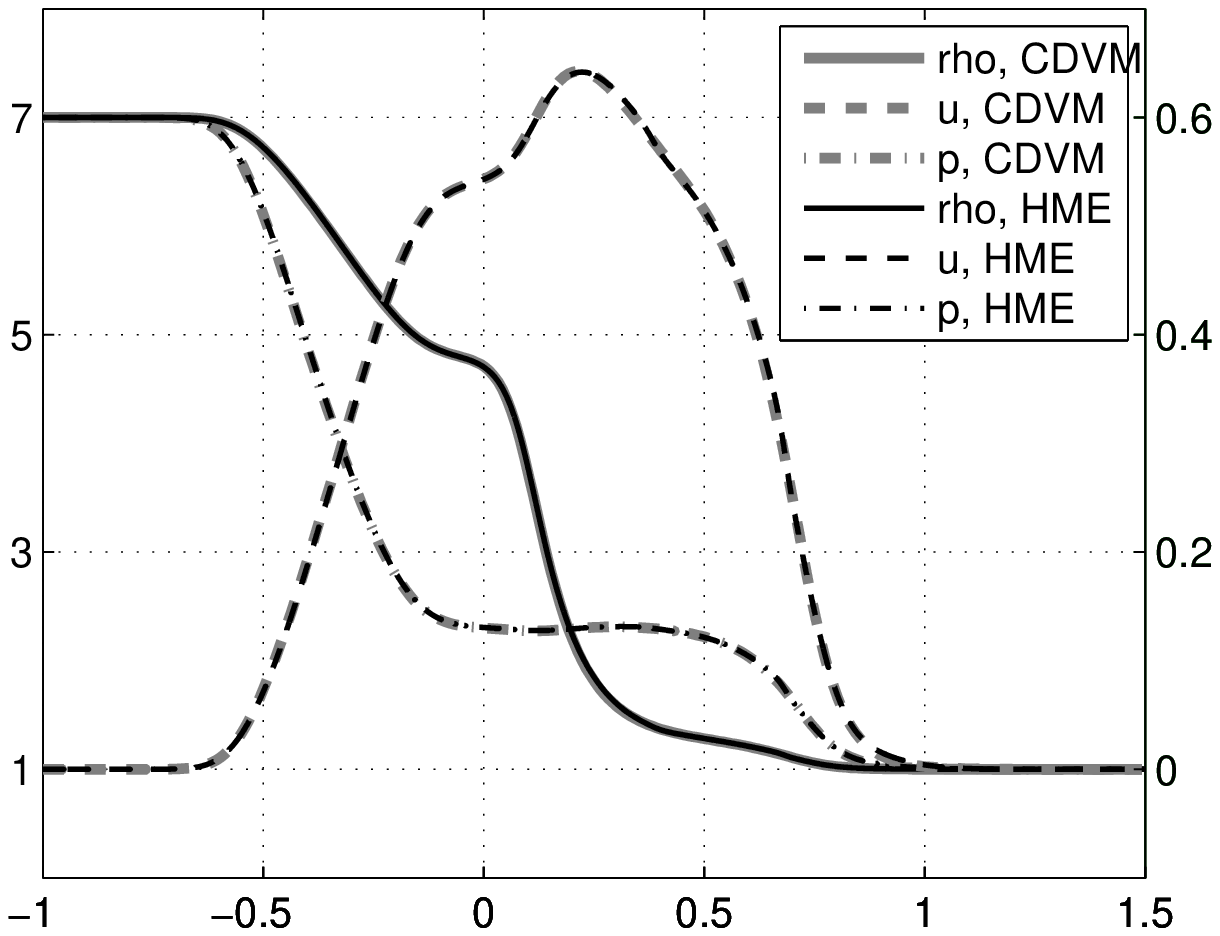}
}
\subfigure[$M=8$]{
\includegraphics[width=.31\textwidth,clip]{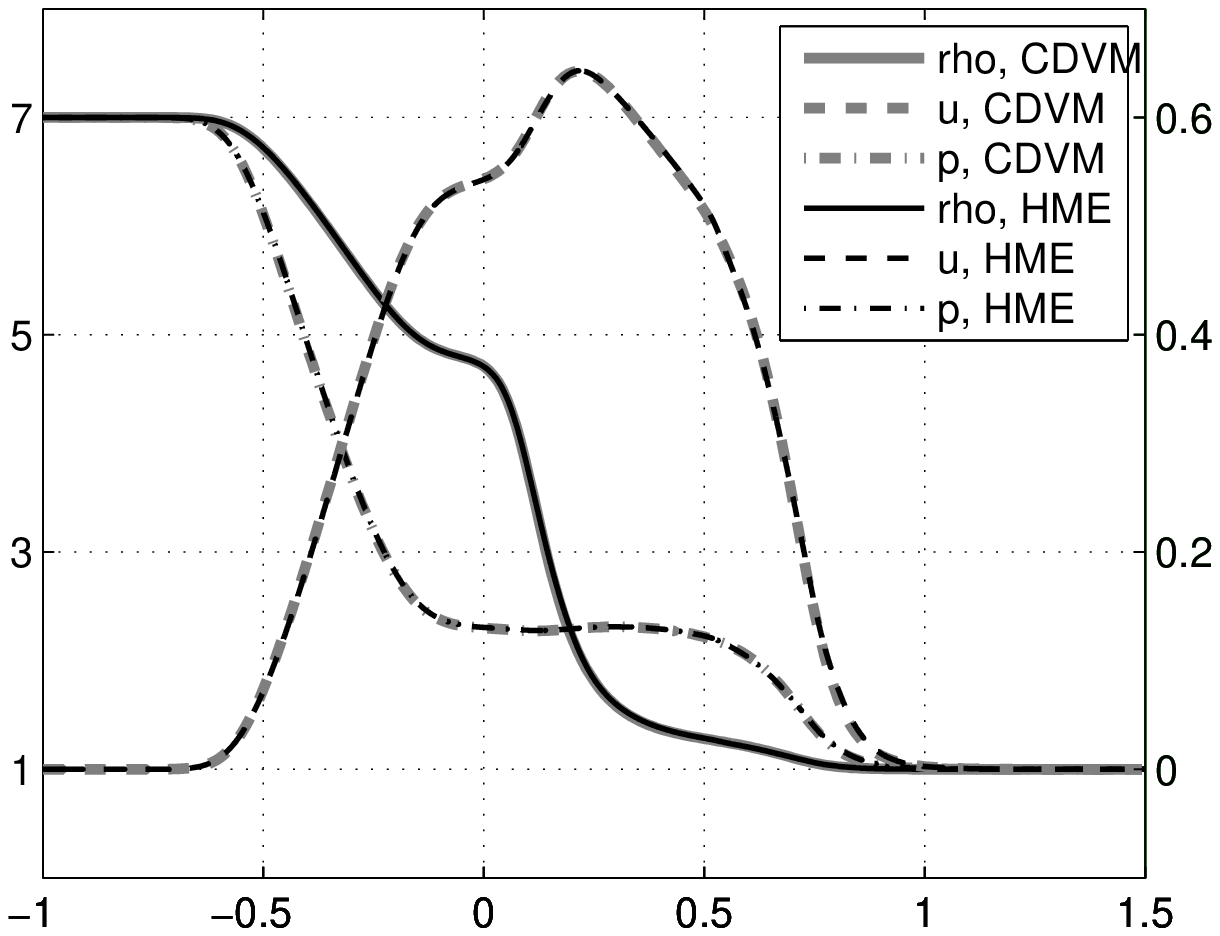}
}
\subfigure[$M=9$]{
\includegraphics[width=.31\textwidth,clip]{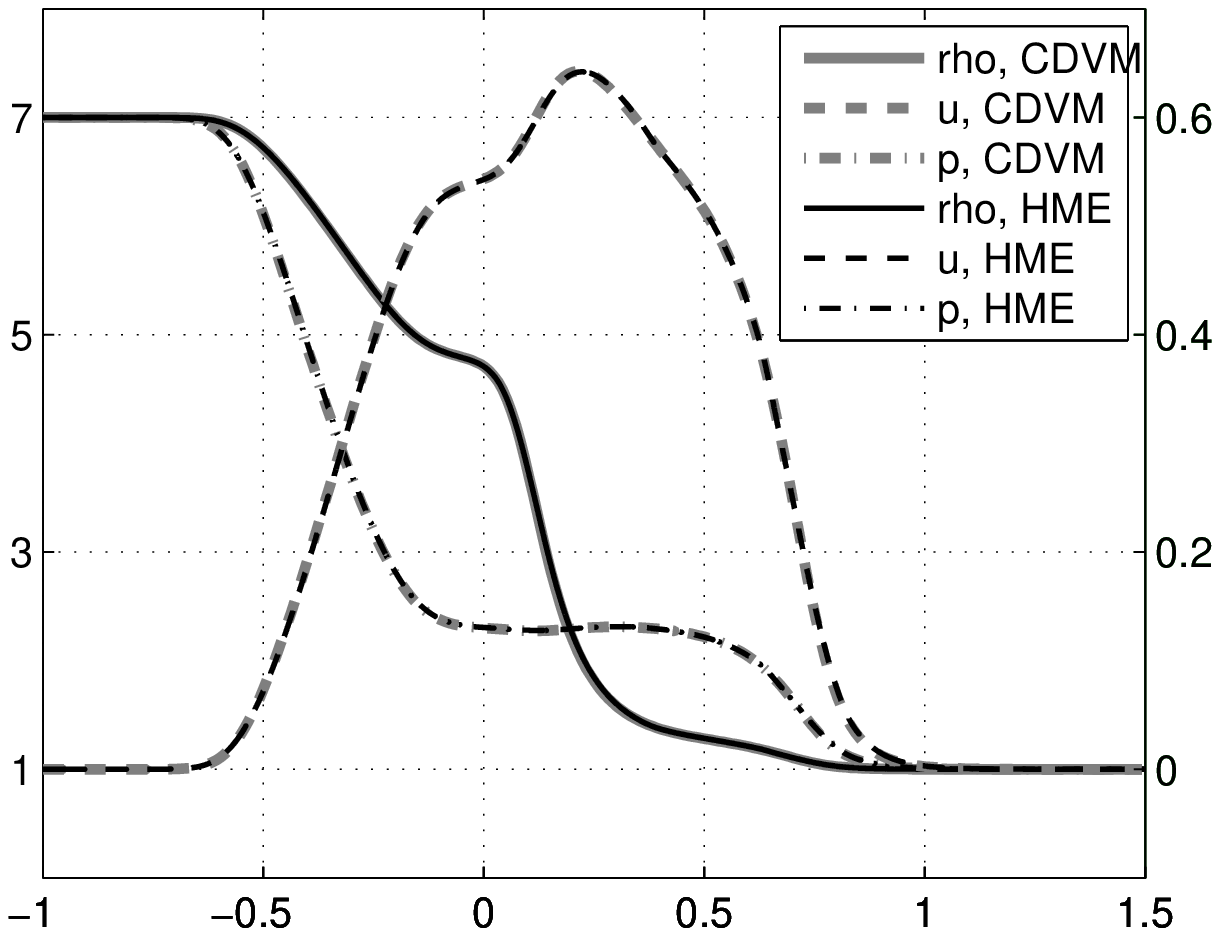}
}
\subfigure[$M=10$]{
\includegraphics[width=.31\textwidth,clip]{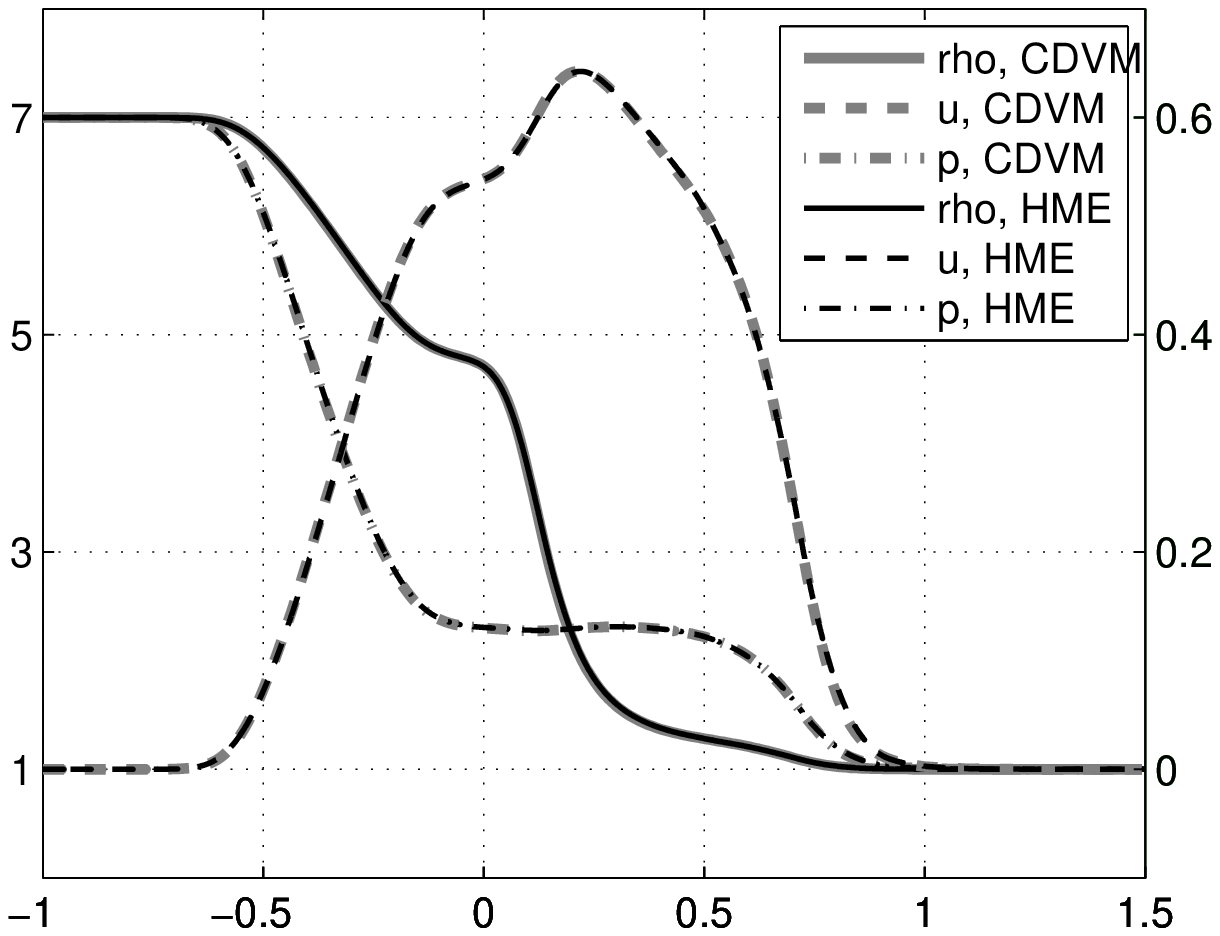}
}
\caption{Numerical results of the shock tube problem for $\Kn = 0.05$.
The left $y$-axis is for $\rho$ and $p$, and the right $y$-axis is for
$u$.}
\label{fig:Kn=0.05}
\end{figure}

For a larger Knudsen number $\Kn = 0.5$, the results are shown in
Figure \ref{fig:Kn=0.5}. These results can also be considered as the
solutions at $t = 0.03$ in the case of $\Kn = 0.05$ (with proper
scaling in the $x$ direction). Thus these actually show the start-up
phases of a shock tube by moment approximation. The discontinuities
are clear for all choices of $M$, and the convergence can also be
observed by eyes.
\begin{figure}[!ht]
\psfrag{rho, CDVM}{\scalebox{.5}{$\rho$, DVM}}
\psfrag{rho, HME}{\scalebox{.5}{$\rho$, HME}}
\psfrag{u, CDVM}{\scalebox{.5}{$u$, DVM}}
\psfrag{u, HME}{\scalebox{.5}{$u$, HME}}
\psfrag{p, CDVM}{\scalebox{.5}{$p$, DVM}}
\psfrag{p, HME}{\scalebox{.5}{$p$, HME}}
\subfigure[$M=2$ (Euler)]{
\includegraphics[width=.31\textwidth,clip]{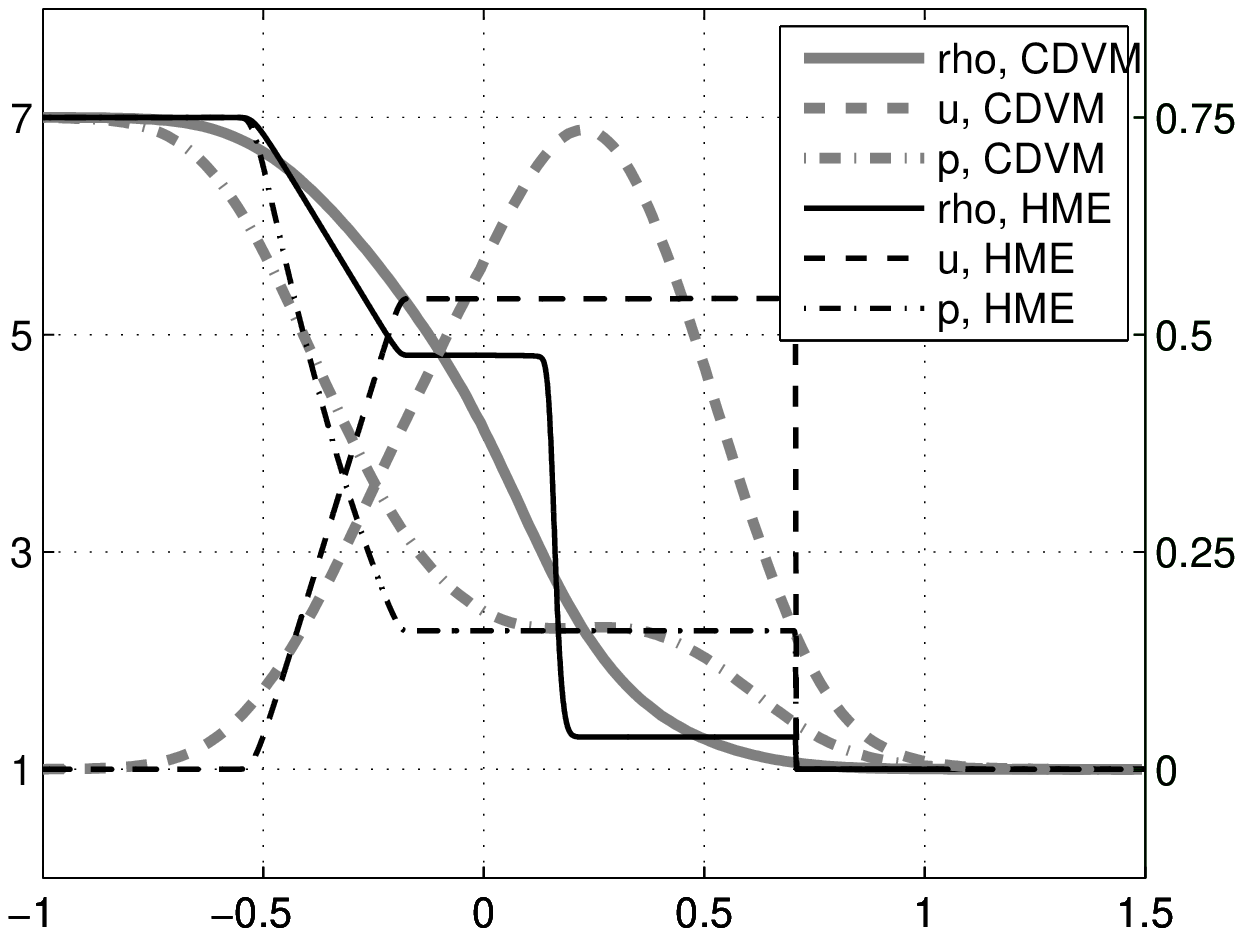}
}
\subfigure[$M=3$]{
\includegraphics[width=.31\textwidth,clip]{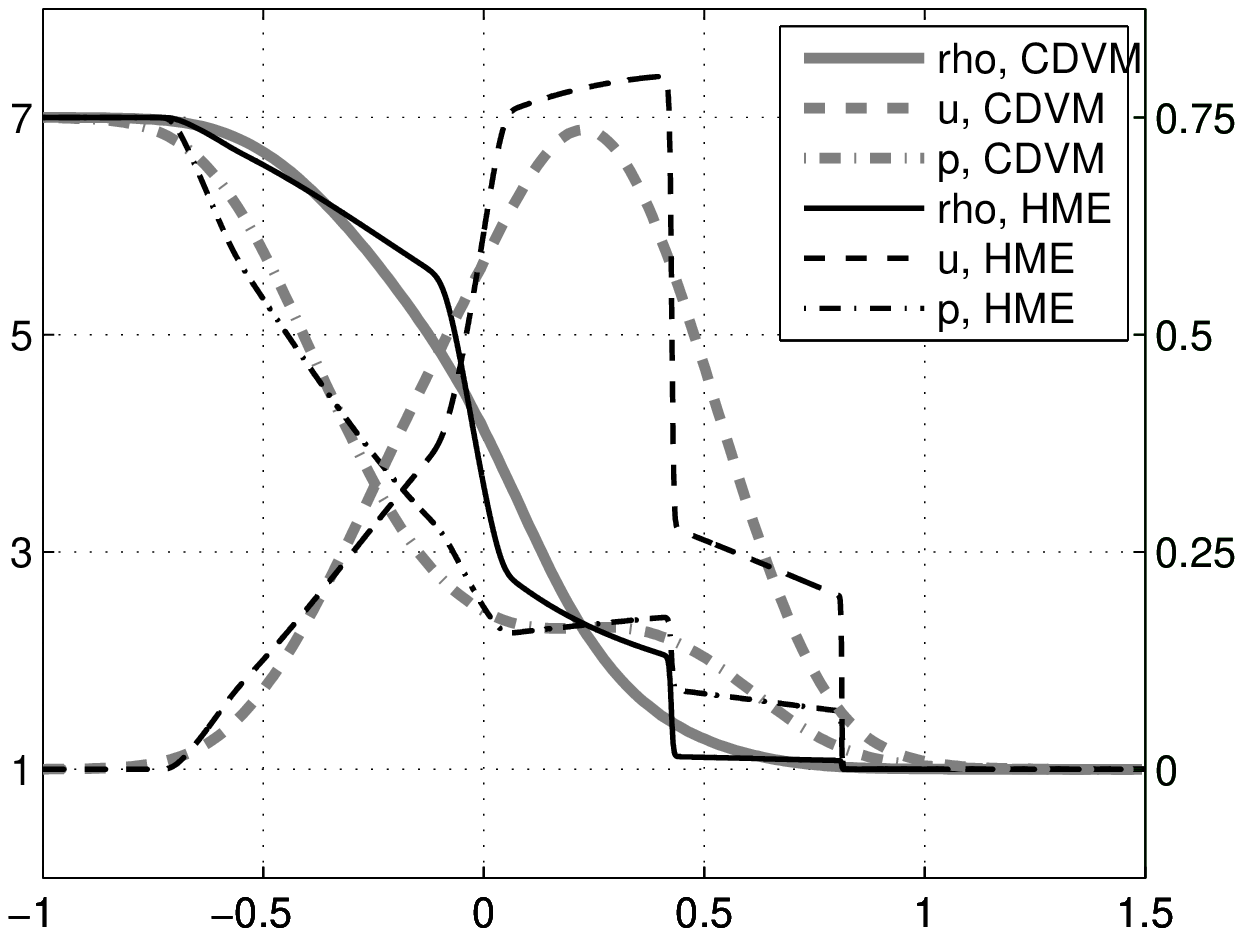}
}
\subfigure[$M=4$]{
\includegraphics[width=.31\textwidth,clip]{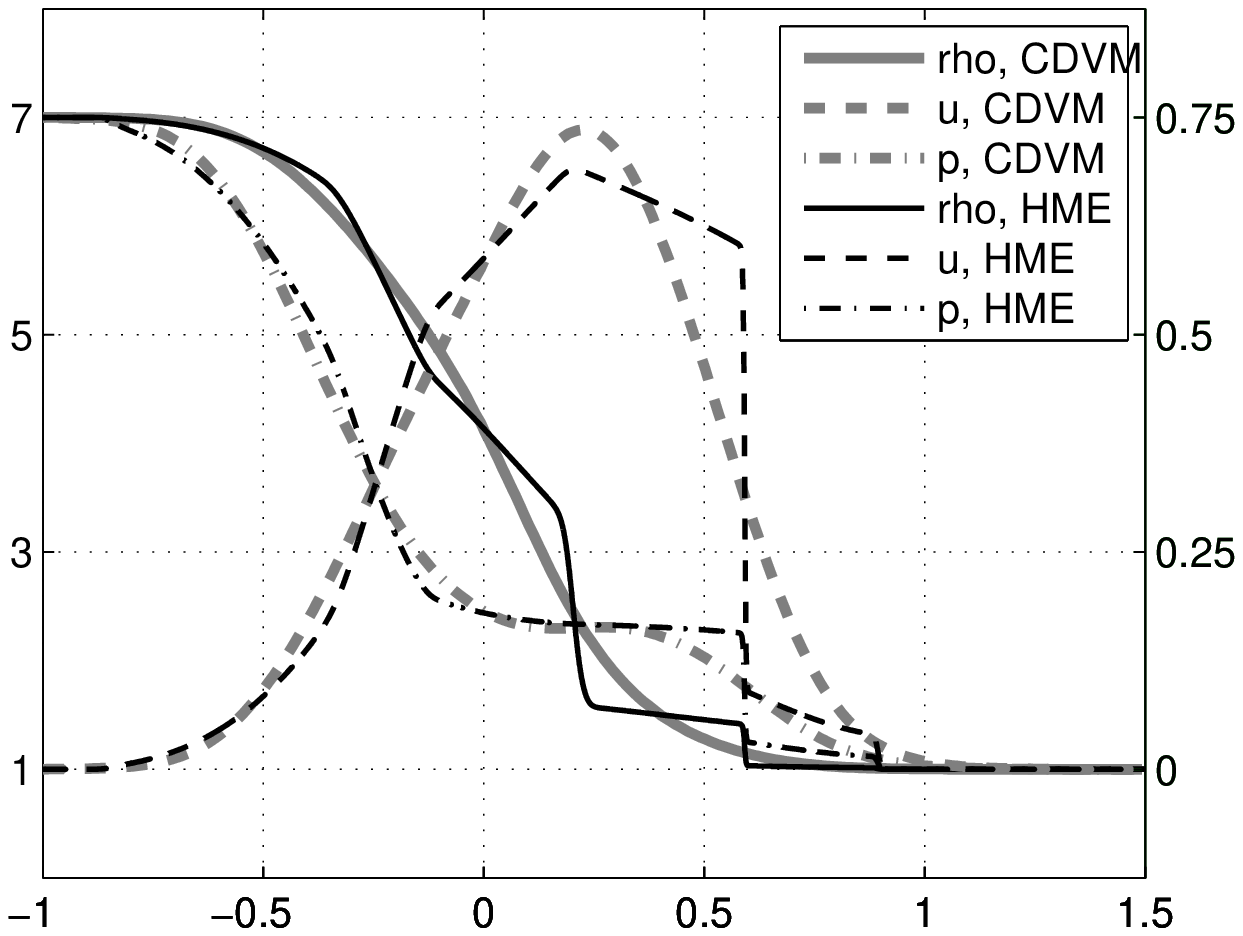}
}
\subfigure[$M=5$]{
\includegraphics[width=.31\textwidth,clip]{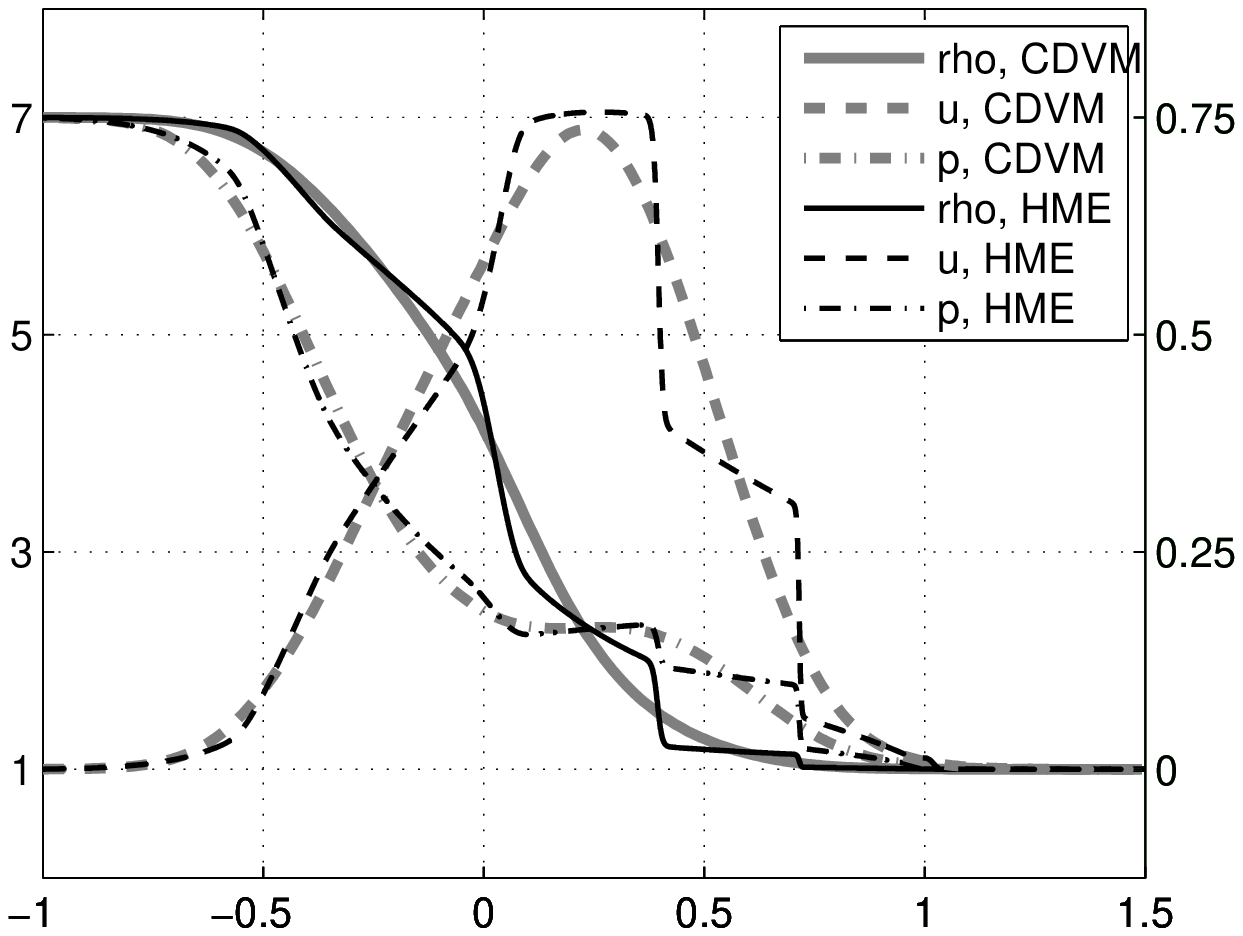}
}
\subfigure[$M=6$]{
\includegraphics[width=.31\textwidth,clip]{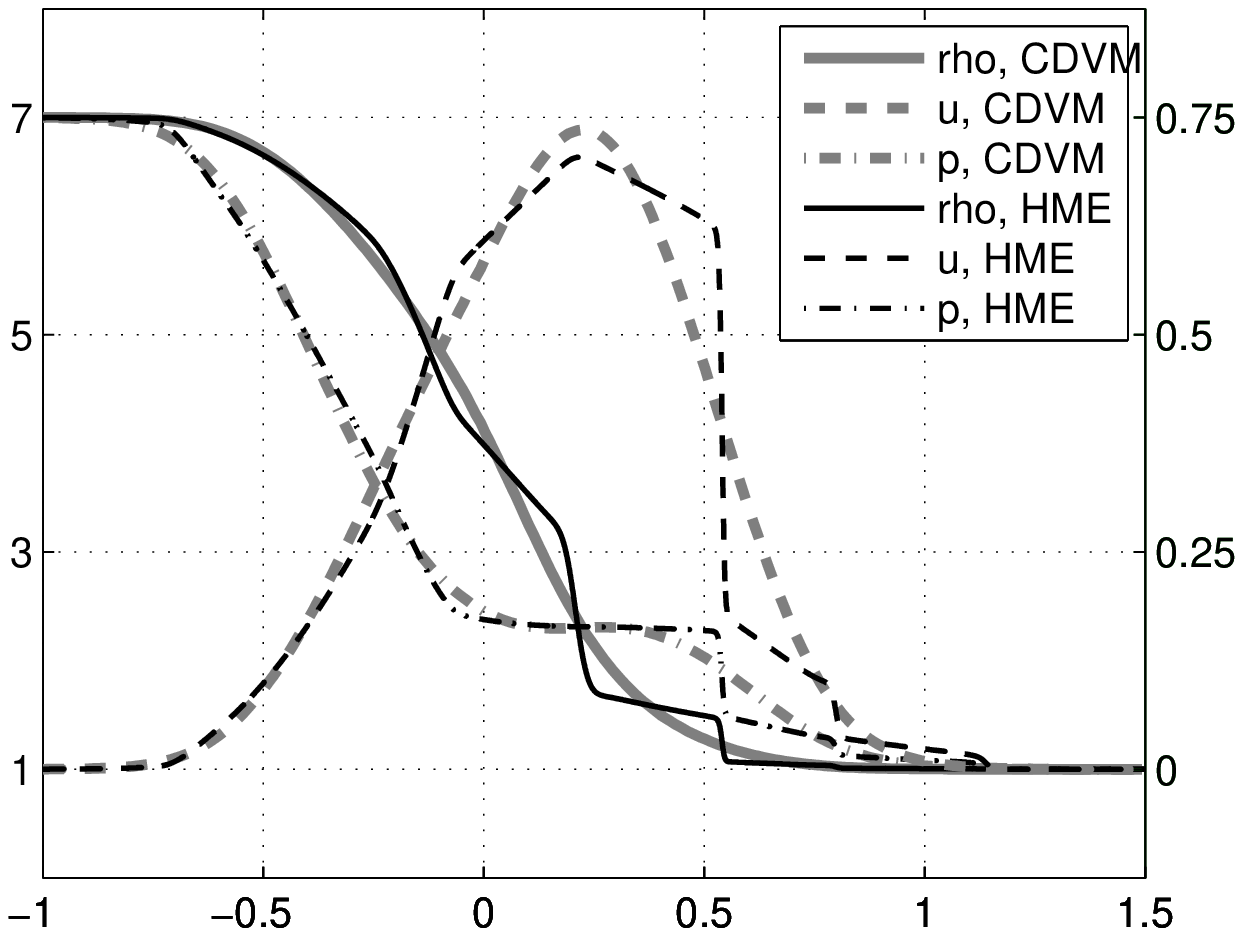}
}
\subfigure[$M=7$]{
\includegraphics[width=.31\textwidth,clip]{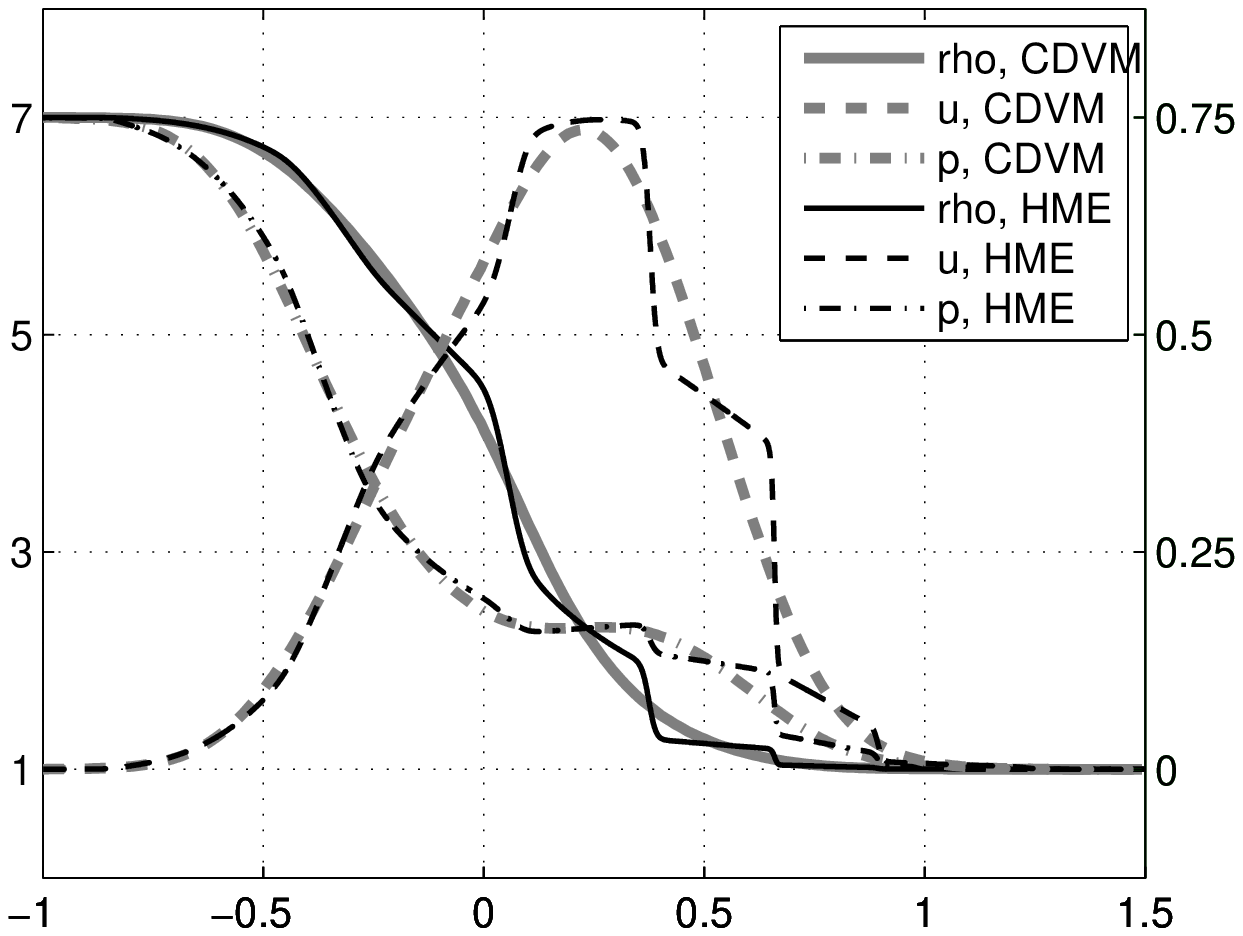}
}
\subfigure[$M=8$]{
\includegraphics[width=.31\textwidth,clip]{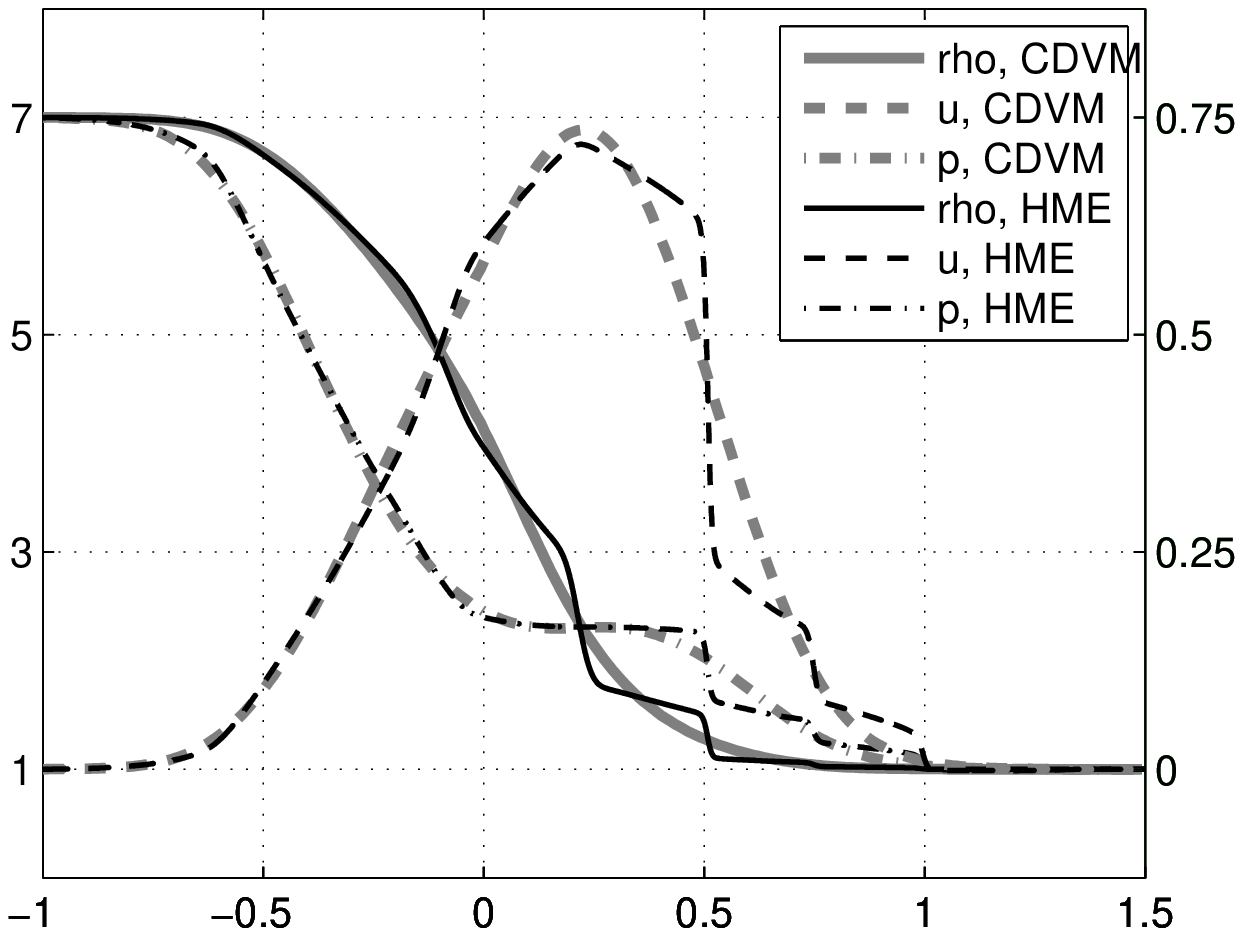}
}
\subfigure[$M=9$]{
\includegraphics[width=.31\textwidth,clip]{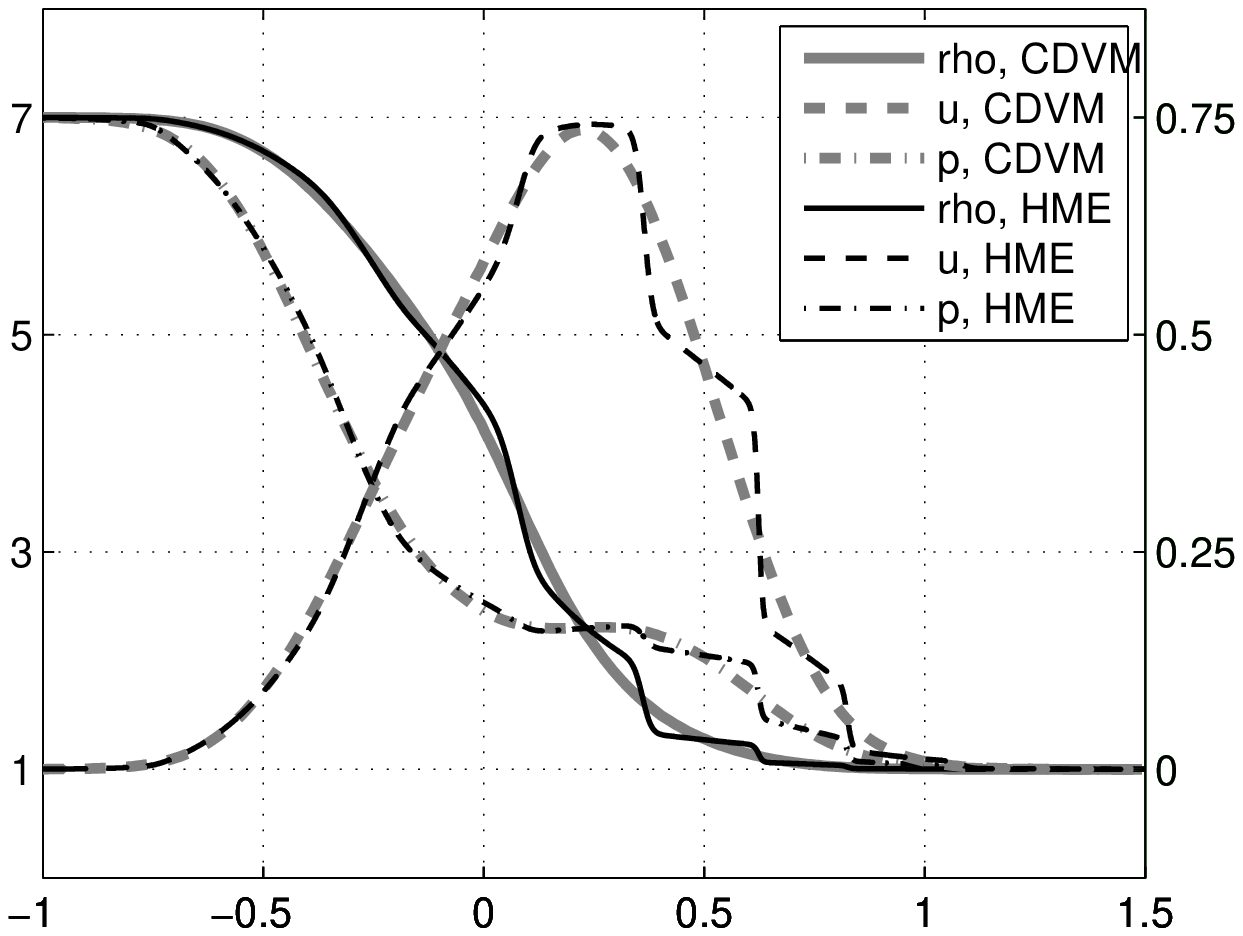}
}
\subfigure[$M=10$]{
\includegraphics[width=.31\textwidth,clip]{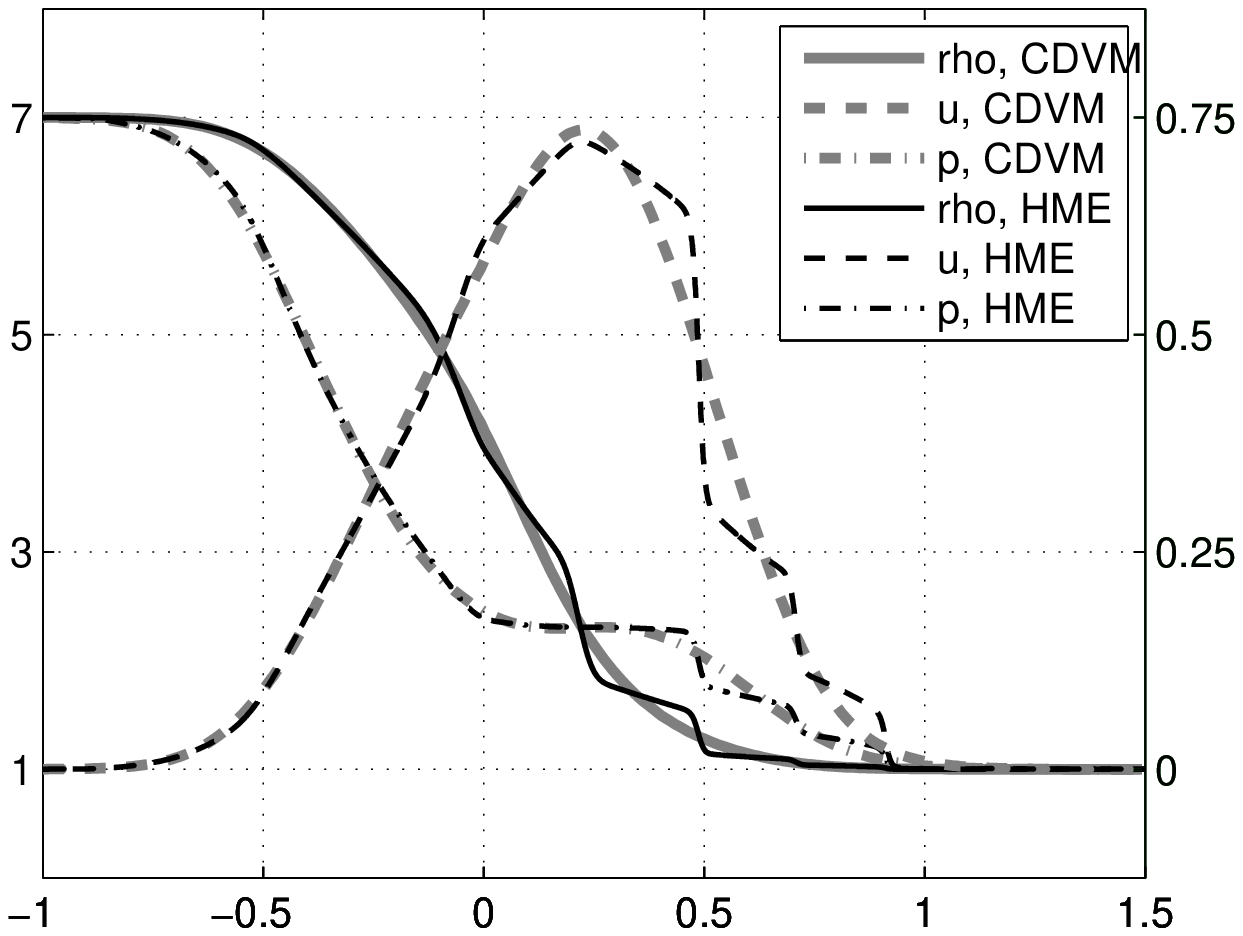}
}
\caption{Numerical results of the shock tube problem for $\Kn = 0.5$.
The left $y$-axis is for $\rho$ and $p$, and the right $y$-axis is for
$u$.}
\label{fig:Kn=0.5}
\end{figure}

\section{Concluding remarks} \label{sec:conclusion} 

We regularize the 1D Grad's moment system to achieve global
hyperbolicity for arbitary order expansion. Fully investigations to
the characteristic waves show that this set of equations may be a
natural extension of Euler equations. Actually, the approach in this
paper has been extended to two or three dimensional Grad's moment
system, and the result is reported in a following paper.

\section*{Acknowledgements}
This research was supported in part by the National Basic Research
Program of China (2011CB309704), the National Science Foundation of
China under grant 10731060 and NCET in China.


\bibliographystyle{plain}
\bibliography{../article}

\begin{thebibliography}{10}

\bibitem{Weiss}
J.~D. Au, M.~Torrilhon, and W.~Weiss.
\newblock The shock tube study in extended thermodynamics.
\newblock {\em Phys. Fluids}, 13(8):2423--2432, 2001.

\bibitem{BGK}
P.~L. Bhatnagar, E.~P. Gross, and M.~Krook.
\newblock A model for collision processes in gases. {I}. small amplitude
  processes in charged and neutral one-component systems.
\newblock {\em Phys. Rev.}, 94(3):511--525, 1954.

\bibitem{Bobylev}
A.~V. Bobylev.
\newblock The {C}hapman-{E}nskog and {G}rad methods for solving the {B}oltzmann
  equation.
\newblock {\em Sov. Phys. Dokl.}, 27(1):29--31, 1982.

\bibitem{Brini}
F.~Brini.
\newblock Hyperbolicity region in extended thermodynamics with 14 moments.
\newblock {\em Continuum Mech. Thermodyn.}, 13(1):1--8, 2001.

\bibitem{Burnett}
D.~Burnett.
\newblock The distribution of velocities in a slightly non-uniform gas.
\newblock {\em Proc. London Math. Soc.}, 39(1):385--430, 1935.

\bibitem{NRxx}
Z.~Cai and R.~Li.
\newblock Numerical regularized moment method of arbitrary order for
  {B}oltzmann-{BGK} equation.
\newblock {\em SIAM J. Sci. Comput.}, 32(5):2875--2907, 2010.

\bibitem{Li}
Z.~Cai, R.~Li, and Z.~Qiao.
\newblock {\NRxx} simulation of microflows with {S}hakhov model.
\newblock {\em SIAM J. Sci. Comput.}, 34(1):A339--A369, 2012.

\bibitem{NRxx_new}
Z.~Cai, R.~Li, and Y.~Wang.
\newblock Numerical regularized moment method for high {M}ach number flow.
\newblock {\em Commun. Comput. Phys.}, 11(5):1415--1438, 2012.

\bibitem{Grad}
H.~Grad.
\newblock On the kinetic theory of rarefied gases.
\newblock {\em Comm. Pure Appl. Math.}, 2(4):331--407, 1949.

\bibitem{Jin}
S.~Jin and M.~Slemrod.
\newblock Regularization of the {B}urnett equations via relaxation.
\newblock {\em J. Stat. Phys}, 103(5--6):1009--1033, 2001.

\bibitem{Junk}
M.~Junk.
\newblock Domain of definition of {L}evermore's five-moment system.
\newblock {\em J. Stat. Phys.}, 93(5):1143--1167, 1998.

\bibitem{Levermore}
C.~D. Levermore.
\newblock Moment closure hierarchies for kinetic theories.
\newblock {\em J. Stat. Phys.}, 83(5--6):1021--1065, 1996.

\bibitem{Maso}
G.~Dal Maso, P.~G. LeFloch, and F.~Murat.
\newblock Definition and weak stability of nonconservative products.
\newblock {\em J. Math. Pures Appl.}, 74(6):483--548, 1995.

\bibitem{Mieussens}
L.~Mieussens.
\newblock Discrete velocity model and implicit scheme for the {BGK} equation of
  rarefied gas dynamics.
\newblock {\em Math. Models Methods Appl. Sci.}, 10(8):1121--1149, 2000.

\bibitem{Muller}
I.~M{\"u}ller and T.~Ruggeri.
\newblock {\em Rational Extended Thermodynamics, Second Edition}, volume~37 of
  {\em Springer tracts in natural philosophy}.
\newblock Springer-Verlag, New York, 1998.

\bibitem{Rhebergen}
S.~Rhebergen, O.~Bokhove, and J.~J.~W. van~der Vegt.
\newblock Discontinuous {G}alerkin finite element methods for hyperbolic
  nonconservative partial differential equations.
\newblock {\em J. Comput. Phys.}, 227(3):1887--1922, 2008.

\bibitem{Shen}
J.~Shen and T.~Tang.
\newblock {\em Spectral and High-Order Methods with Applications}, volume~3 of
  {\em Mathematics Monograph Series}.
\newblock Science Press, Beijing, P. R. China, 2006.

\bibitem{Struchtrup}
H.~Struchtrup.
\newblock {\em Macroscopic Transport Equations for Rarefied Gas Flows:
  Approximation Methods in Kinetic Theory}.
\newblock Springer, 2005.

\bibitem{Struchtrup2003}
H.~Struchtrup and M.~Torrilhon.
\newblock Regularization of {G}rad's 13 moment equations: Derivation and linear
  analysis.
\newblock {\em Phys. Fluids}, 15(9):2668--2680, 2003.

\bibitem{Tallec}
P.~L. Tallec and J.~P. Perlat.
\newblock Numerical analysis of {L}evermore's moment system.
\newblock Rapport de recherche 3124, INRIA Rocquencourt, March 1997.

\bibitem{Torrilhon2000}
M.~Torrilhon.
\newblock Characteristic waves and dissipation in the 13-moment-case.
\newblock {\em Continuum Mech. Thermodyn.}, 12(5):289--301, 2000.

\bibitem{Torrilhon2006}
M.~Torrilhon.
\newblock Two dimensional bulk microflow simulations based on regularized
  {G}rad's 13-moment equations.
\newblock {\em SIAM Multiscale Model. Simul.}, 5(3):695--728, 2006.

\bibitem{Torrilhon2010}
M.~Torrilhon.
\newblock Hyperbolic moment equations in kinetic gas theory based on
  multi-variate {P}earson-{IV}-distributions.
\newblock {\em Commun. Comput. Phys.}, 7(4):639--673, 2010.

\bibitem{Au}
M.~Torrilhon, J.~Au, D.~Reitebuch, and W.~Weiss.
\newblock The {R}iemann-problem in extended thermodynamics.
\newblock In H.~Freistu{\"u}hler and G.~Warnecke, editors, {\em Hyperbolic
  Problems: Theory, Numerics, Applications, Vols {I} and {II}}, volume 140 of
  {\em International series of numerical mathematics}, pages 79--88.
  Birkh{\"a}user, 2001.

\end{thebibliography}
\end{document}